\documentclass{lmcs}


\AtBeginDocument{%
}

\usepackage[T1]{fontenc}

\usepackage{booktabs}
\usepackage{amssymb}
\usepackage{amsmath}
\usepackage{ebproof}
\usepackage{stmaryrd}
\usepackage{xspace}

\newcommand{\NN}{\mathbb{N}}
\newcommand{\ZZ}{\mathbb{Z}}
\renewcommand{\AA}{\mathcal{A}}
\newcommand{\BB}{\mathcal{B}}
\newcommand{\DD}{\mathcal{D}}
\newcommand{\FF}{\mathcal{F}}
\newcommand{\GG}{\mathcal{G}}
\newcommand{\PP}{\mathcal{P}}
\newcommand{\RR}{\mathcal{R}}
\newcommand{\UU}{\mathcal{U}}
\newcommand{\TT}{\mathcal{T}}
\newcommand{\VV}{\mathcal{V}}

\newcommand{\m}[1]{\mathsf{#1}}

\newcommand{\seq}[2][n]{{#2_1},\dots,{#2_{#1}}}

\newcommand{\hole}{\square}
\newcommand{\rt}{\mathrm{root}}
\newcommand{\Var}{{\mathcal{V}\mathsf{ar}}}
\newcommand{\Ce}{\mathcal{C}_\mathcal{E}}
\newcommand{\DP}{\mathsf{DP}}
\newcommand{\idinc}{\m{h}}

\newcommand{\superterm}{\trianglerighteqslant}
\newcommand{\prsuperterm}{\rhd}
\newcommand{\nprsuperterm}{\ntriangleright}
\newcommand{\xto}[1]{\xrightarrow{\smash{#1}}}
\newcommand{\rto}{\xto{\epsilon}}

\newcommand{\lex}{\mathsf{lex}}

\newcommand{\lexprod}{\otimes}

\newcommand{\spo}{\mathsf{S}}
\newcommand{\mspo}{\mathsf{MS}}
\newcommand{\wpo}{\mathsf{W}}
\newcommand{\gwpo}{\mathsf{G}}
\newcommand{\mgwpo}{\mathsf{MG}}
\newcommand{\kbo}{\mathsf{K}}

\newcommand{\Dershowitz}{\mathsf{D}}
\newcommand{\SN}{\mathsf{SN}}

\newcommand{\WPO}[1]{\hyperref[wpo:#1]{#1}}
\newcommand{\GWPO}[1]{\hyperref[gwpo:#1]{#1}}
\newcommand{\SPO}[1]{\hyperref[spo:#1]{#1}}
\newcommand{\DER}[1]{\hyperref[dershowitz:#1]{#1}}
\newcommand{\KBOGWPO}[1]{\hyperref[kbogwpo:#1]{#1}}
\newcommand{\KBOSPO}[1]{\hyperref[kbospo:#1]{#1}}
\newcommand{\WPORP}[1]{\hyperref[wpoRP:#1]{#1}}
\newcommand{\GWPORP}[1]{\hyperref[gwpoRP:#1]{#1}}
\newcommand{\SPORP}[1]{\hyperref[spoRP:#1]{#1}}

\newcommand{\LinPol}{\texttt{L}}
\newcommand{\MaxPlus}{\texttt{M}}
\newcommand{\Finite}{\texttt{F}_2}
\newcommand{\Trivial}{\texttt{T}}
\newcommand{\Prec}{\texttt{P}}
\newcommand{\GWPOtt}{\texttt{GWPO}}
\newcommand{\SPOtt}{\texttt{SPO}}
\newcommand{\WPOtt}{\texttt{WPO}}

\newcommand{\emptylist}{[\,]}
\renewcommand{\emptyset}{\varnothing}
\newcommand{\idrel}{=}

\newcommand{\AProVE}{\textsf{AProVE}\xspace}
\newcommand{\NaTT}{\textsf{NaTT}\xspace}
\newcommand{\Mint}{\textsf{Mint}\xspace}

\begin{document}

\title{Semantic Path Orders are Generalized~Weighted~Path~Orders}
\thanks{This research was supported by JSPS KAKENHI Grant Numbers
JP22K11900 and JP25KJ1363.}

\author[T.~Saito]{Teppei Saito\lmcsorcid{https://orcid.org/0009-0001-9786-0044}}
\author[N.~Hirokawa]{Nao Hirokawa\lmcsorcid{https://orcid.org/0000-0002-8499-0501}}

\address{JAIST, Japan}
\email{saito@jaist.ac.jp, hirokawa@jaist.ac.jp}

\begin{abstract}
\noindent
Monotonic semantic path orders (MSPOs) are powerful reduction orders, but their
effective use in theorem proving and termination analysis raises three issues:
efficient comparison of terms, ground totality, and a reduction-pair formulation.
Inspired by weighted path orders (WPOs),
we introduce generalized weighted path orders (GWPOs)
which encompass both WPOs and MSPOs.
GWPOs and MSPOs coincide as classes of reduction orders,
thereby yielding efficient and ground-total subclasses of MSPOs.
We further formulate the reduction-pair version of the GWPO for the dependency pair method,
from which a reduction-pair version of the MSPO is obtained.
In contrast to the reduction-order setting, the correspondence between GWPOs and MSPOs no longer fully carries over,
establishing the GWPO as a framework of independent interest.
Effectiveness of the resulting reduction pairs is assessed by
standard termination benchmarks.
\end{abstract}

\maketitle

\section{Introduction}
\label{S:introduction}
Term orders are a fundamental tool in termination analysis of term
rewrite systems, and they also underlie completion-based automated theorem
proving. Classical examples include
the Knuth--Bendix order (KBO)~\cite{KB70} and the lexicographic path order (LPO)~\cite{KL80},
along with numerous variants and generalizations.
Typically, these orders are used as so-called \emph{reduction orders} $>$,
which characterize termination of term rewrite systems $\RR$
by orientability 
$\ell > r$ of all rules $\ell \to r \in \RR$.

The \emph{semantic path order} (SPO) was introduced by Kamin and L\'{e}vy alongside the LPO
in their unpublished but celebrated note~\cite{KL80}.
It generalizes the precedence of the LPO (used to compare 
function symbols) to relations on terms, thereby
allowing semantic ingredients to be used in term comparison.
Despite its expressive power for termination analysis, the SPO does not necessarily satisfy the
monotonicity property required of reduction orders.
Building on Geser's work~\cite{Ges92,Ges96},
Borralleras, Ferreira, and Rubio~\cite{BFR00,Bor03} introduced a reduction-order version of the SPO
dubbed the \emph{monotonic semantic path order} (MSPO).
Their work paved the way to automated termination analysis with the SPO~\cite{BR07}.

The motivation of this article lies in making the (M)SPO more effective in termination analysis and theorem proving.
This raises the following three questions.
\begin{itemize}
\item Under what conditions is an MSPO 
\emph{ground total}?
This property means that ground terms are totally ordered, and it
is a standard requirement for the refutational completeness of theorem proving procedures,
such as ordered completion~\cite{BDP89} and superposition calculus~\cite{BG94}.
\item Can terms be compared by an MSPO 
\emph{efficiently}? This is also important for theorem proving,
where term comparisons are performed frequently.
For classical orders such as the KBO and the LPO,
efficient comparison algorithms have been studied, for instance, in ~\cite{Loc06KBO,Loc06LPO}.
\item Can MSPOs be formulated as \emph{reduction pairs} (a more general form of reduction orders)?
Such a formulation (based on pairs of quasi-orders and well-founded orders on terms)
is needed for their effective use in termination proving based on the dependency pair method~\cite{AG00}.
\end{itemize}

The technical key lies in the \emph{weighted path order} (WPO)~\cite{YKS15},
another generalization of the LPO and also the KBO\@.
It combines a semantic comparison based on an interpretation (e.g., a polynomial interpretation~\cite{Lan75})
with a syntactic comparison based on a precedence.
Regarding the three questions above, more is known about the WPO than about the MSPO\@.
First, there is a reasonable sufficient condition for a WPO to be a ground-total reduction order.
Second, the WPO admits a subclass with an efficient implementation, namely the KBO~\cite{Loc06KBO}.
Third, WPOs have been formulated as reduction pairs~\cite{YKS15} and used successfully in automated termination provers~\cite{YKS14,Sch20}.
In our previous work~\cite{SH23}, we showed that WPOs formulated as reduction orders
can be represented as MSPOs. This connection suggests that the
aforementioned features of the WPO might be transferred to the MSPO.

To this end, we introduce a generalization of the WPO 
dubbed the \emph{generalized weighted path order} (GWPO).
The definition shares the same recursion structure as the WPO,
but the parameters of the WPO are generalized from an algebra and a precedence
to reduction triples, which are basic ingredients of the MSPO~\cite{Bor03}.
By definition, the WPO and the (M)SPO are immediate special cases of the GWPO.

Using the GWPO, the issues of ground totality and efficiency are addressed as follows.  
As with the WPO, there is a reasonable sufficient condition for a GWPO to be a ground-total reduction order,
and the GWPO admits an efficient subclass analogous to the KBO\@.
These results can then be transferred to the MSPO by using a simulation result 
that every GWPO formulated as a reduction order can be represented as an MSPO\@.
As a byproduct, this implies that the GWPO and the MSPO coincide as classes of reduction orders.

Following the WPO, we also formulate GWPOs as reduction pairs.
In this formulation, WPOs and (M)SPOs again arise as immediate special cases.
Interestingly, the aforementioned simulation result for reduction orders does not fully carry over to the reduction-pair setting,
which implies that the GWPO is of independent interest.
We also evaluate the resulting reduction pairs on standard benchmark problems for automated termination proving~\cite{TPDB}.

This article is based largely on the first author's PhD thesis~\cite{Sai26},
and substantially extends our conference paper~\cite{SH23}.
More precisely, the main result of the conference paper that WPOs are (M)SPOs is a special case of
\autoref{T:gwpo-spo}, which states that GWPOs are SPOs.
The other results are new with respect to the conference paper.

The remainder of the paper is organized as follows.
\autoref{S:preliminaries} collects basic facts about orders, term rewriting, and interpretation methods.
The next two sections are devoted to the main results of this article,
and begin with a brief overview of their subsections.
\autoref{S:reduction-order} is devoted to GWPOs as reduction orders,
and \autoref{S:reduction-pair} to GWPOs as reduction pairs.
The article concludes with \autoref{S:conclusion}, which discusses related and future work.

\section{Preliminaries}
\label{S:preliminaries}

\subsection{Orders and Order Pairs}
A strict order $>$ on a set $X$ is called a \emph{well-founded order} if it is well-founded, i.e.,
there is no infinite descending sequence $x_0 > x_1 > \cdots$.
A quasi-order $\geqslant$ on a set $X$ is \emph{well-founded} if its strict part is well-founded.
For example, the standard strict order $>_\NN$ and quasi-order $\geqslant_\NN$ on natural numbers $\NN$ are well-founded.

A pair $({\geqslant}, {>})$ of relations on the same set $X$ is \emph{compatible}
if ${> \cdot \geqslant} \subseteq {>}$ and ${\geqslant \cdot >} \subseteq {>}$.
Here $R \cdot S$ denotes the \emph{composition} of $R$ and $S$ defined as
$\{ (x,z) \mid \text{$(x,y) \in R$ and $(y,z) \in S$} \}$.
The pair $({\geqslant},{>})$ is an \emph{order pair} if it is compatible, 
$\geqslant$ is a quasi-order and $>$ is a strict order;
it is \emph{well-founded} if $>$ is well-founded; and
it is \emph{normal} if $> {\subseteq} \geqslant$.\footnote{%
We note that there are some variations in the definition of an order pair.
For instance, what is called an order pair in \cite[Definition~1]{Yam22IJCAR} corresponds to
a normal order pair in our terminology.}
For example, $({\geqslant_\NN}, {>_\NN})$ is a well-founded and normal order pair.
We say that a relation $R$ \emph{extends} another relation $S$ if $S \subseteq R$.
Similarly, we say that an order pair $({\geqslant},{>})$ extends another one $({\sqsupseteq}, {\sqsupset})$ on $X$
if $\geqslant$ extends $\sqsupseteq$ and $>$ extends $\sqsupset$.
We call $({X \times X}, {\emptyset})$ a \emph{trivial} order pair.

Given a set $X$ and a natural number $n$,
we write $X^{n}$ for the set of all tuples $(x_1, \ldots, x_n)$ of length $n$
consisting of elements of $X$.
We also write $X^{*}$ for the set of all tuples of arbitrary finite length over $X$.
The elements of $X^{*}$ are called \emph{lists}.
We recall the \emph{lexicographic extension}, which extends an order pair to lists.

\begin{defi}
\label{D:lex-ext}
Let $({\geqslant}, {>})$ be a pair of relations on a set $A$.  
The \emph{lexicographic extension}
$({\geqslant^\lex}, {>^\lex})$ is a pair of relations on $A^{*}$ defined recursively:
$(\seq[m]{a}) \geqslant^{\lex} (\seq[n]{b})$ if one of the following cases holds.
\begin{itemize}
    \item $n = 0$
    \item $m, n > 0$ and $a_1 > b_1$
    \item $m, n > 0$ and $a_1 \geqslant b_1$, and moreover $(a_2, \ldots, a_m) \geqslant^\lex (b_2, \ldots, b_n)$
\end{itemize}
Similarly, $(\seq[m]{a}) >^{\lex} (\seq[n]{b})$ if one of the following cases holds. 
\begin{itemize}
    \item $m > 0$ and $n = 0$
    \item $m, n > 0$ and $a_1 > b_1$
    \item $m, n > 0$ and $a_1 \geqslant b_1$, and moreover $(a_2, \ldots, a_m) >^\lex (b_2, \ldots, b_n)$
\end{itemize}   
\end{defi}

It is known that $({\geqslant^\lex}, {>^\lex})$ is an order pair on $A^*$ if $({\geqslant}, {>})$ is.
However, $>^\lex$ is not well-founded even if $>$ is.
This is witnessed by
the lexicographic extension $({\geqslant^\lex}, {>^\lex})$ of $({\geqslant_\NN}, {>_\NN})$,
in which
\(
(1) >^\lex (0, 1) >^\lex (0, 0, 1) >^\lex \cdots
\)
holds.
Well-foundedness is recovered by imposing a fixed length bound on $A^*$,
see \cite[Section~3]{ST13} for a related discussion.
Given a set $A$,
we write $A^{\leqslant n}$ for the union of $A^{i}$ for all natural numbers $i \leqslant n$.

\begin{prop}
\label{P:lex-ext}
If $({\geqslant}, {>})$ is a well-founded order pair,
then the restriction of $>^\lex$ to $A^{\leqslant n}$ for a natural number $n \in \NN$ is well-founded.
\end{prop}

Quite often, for a single strict order $>$ (rather than an order pair),
we consider the strict part $>^\lex$ of the lexicographic extension of $({\idrel}, {>})$,
where $=$ is the identity relation.
In that case, we may simply call $>^\lex$ the lexicographic extension of $>$.

\subsection{Term Rewriting}
\label{SS:rewriting}

We begin with the definition of terms.
A \emph{signature} $\FF$ is a designated set of function symbols.
Each function symbol is associated with a unique natural number called \emph{arity},
which specifies the number of arguments of the function.
When we need to indicate the arity $n$ of a function symbol $f$,
we write $f^{(n)}$ for $f$.
A function symbol $f^{(n)}$ is called a \emph{constant} if 
$n = 0$,
\emph{unary} if $n = 1$, and
\emph{binary} if $n =2$.
Let $\VV$ be a designated countably infinite set of \emph{variables}
with $\FF \cap \VV = \varnothing$.
The set $\TT(\FF,\VV)$ of \emph{terms} is defined inductively: every variable is a term,
and $f(t_1, \ldots, t_n)$ is a term if $f^{(n)} \in \FF$ and $t_1, \ldots, t_n$ are terms.
The set $\TT(\FF,\VV)$ may be 
written as
$\TT$ when $\FF$ and $\VV$ are clear from the context.

We define a few convenient notations for terms.
For binary function symbols like $+$, we may use infix notation $t_1 + t_2$
to denote the term $+(t_1, t_2)$.
For example, let $\FF = \{ \cdot^{(2)}, \m{S}^{(0)}, \m{K}^{(0)} \}$, and $x, y \in \VV$.
Then we have $(\m{K} \cdot x) \cdot y \in \TT(\FF, \VV)$.
Also, we may use vector notation $t = f(\bar{t})$
to indicate that $t$ is a term headed by the function symbol $f$. 
In this case, $f$ is the \emph{root symbol} of $t$, denoted by $\rt(t)$.
For instance, $\cdot$ is the root symbol of $(\m{K} \cdot x) \cdot y$.
The set of all variables occurring in a term $t$ is denoted by $\Var(t)$.
The \emph{size} $|t|$
of a term $t$ is the number of function symbols and variables occurring in $t$.
For example, we have 
$\Var(t) = \{x,y\}$ and
$|t| = 5$ for $t = (\m{K} \cdot x) \cdot y$.

Let $\hole$ be a fresh constant with $\hole \notin \FF$, which we call the \emph{hole}.
\emph{Contexts} are terms over
$\FF \cup \{ \hole \}$ that contain exactly one hole.  The term resulting
from replacing $\hole$ in a context $C$ by a term $t$ is denoted by $C[t]$.
We write $s \superterm t$ if there is a context $C$ with $s = C[t]$.
In such a case, $t$ is a \emph{subterm} of $s$.
The strict part of $\superterm$ is denoted by $\prsuperterm$.
Similarly, $s \prsuperterm t$ reads as $t$ is a \emph{proper subterm} of $s$.
For example, the term $(\m{K} \cdot x) \cdot y$ has four proper subterms, namely
$\m{K} \cdot x$, $\m{K}$, $x$, and $y$.

A \emph{substitution} is a mapping $\sigma$ from variables to terms such
that $\{ x \in \VV \mid \sigma(x) \neq x \}$ is finite.  The application
$t\sigma$ of a substitution $\sigma$ to a term $t$ is inductively defined
as follows:
$t\sigma = \sigma(t)$ if $t$ is a variable, and 
$t\sigma = f(t_1\sigma,\ldots,t_n\sigma)$ if $t = f(\seq{t})$.
We write $\{ x_1 \mapsto t_1, \ldots, x_k \mapsto t_k \}$ for the substitution $\sigma$
such that $\sigma(y) = t_i$ if $y = x_i$ for some $i$, and $\sigma(y) = y$ otherwise.
For example, $((\m{K} \cdot x) \cdot y) \sigma = (\m{K} \cdot \m{K}) \cdot \m{K}$
for $\sigma = \{ x \mapsto \m{K}, y \mapsto \m{K}\}$.

A pair $(\ell,r)$ of terms is said to be a \emph{rewrite rule} if
$\ell$ is not a variable and every variable in $r$ occurs in $\ell$.
Rewrite rules $(\ell,r)$ are written as $\ell \to r$.  A set of rewrite
rules is called a \emph{term rewrite system} (TRS).
A famous example of a TRS is
combinatory logic, consisting of the rewrite rules
$((\m{S} \cdot x) \cdot y) \cdot z \to (x \cdot z) \cdot (y \cdot z)$ and 
$(\m{K} \cdot x) \cdot y \to x$.
Let $\RR$ be a TRS\@.
The relation $\to_\RR$ is defined on terms as follows: $s \to_\RR t$
if there exist a rule $\ell \to r \in \RR$, a context $C$, and a
substitution $\sigma$ such that $s = C[\ell\sigma]$ and $t = C[r\sigma]$
hold. When $C = \hole$, it is written as $s \rto_\RR t$.
For instance, in the TRS $\RR$ of combinatory logic we have
\(
((\m{S} \cdot \m{K}) \cdot \m{K}) \cdot x
\rto_{\RR} 
(\m{K} \cdot x) \cdot (\m{K} \cdot x)
\rto_{\RR}
x
\).

Let $\hookrightarrow$ be a relation on terms. We say that
$\hookrightarrow$ is \emph{closed under contexts}
or just \emph{monotone} if $C[s] \hookrightarrow C[t]$ holds whenever $s  \hookrightarrow t$ and $C$ is a context;
$\hookrightarrow$ is \emph{closed under substitutions}
or just \emph{stable} if
$s\sigma \hookrightarrow t\sigma$ holds whenever $s \hookrightarrow t$ and $\sigma$ is a substitution;
$\hookrightarrow$ is a \emph{rewrite relation} if it is stable and monotone; and
$\hookrightarrow$ has the \emph{subterm property}
if $s \hookrightarrow t$ whenever $s \prsuperterm t$.
For example, the relation $\to_\RR$ induced by a TRS $\RR$ is by definition closed under contexts and substitutions.
We say that a pair $({\geqslant},{>})$ of relations on terms is \emph{stable} if $\geqslant$ and $>$ are stable.

A term $t$ is \emph{terminating} with respect to a relation
$\hookrightarrow$ on terms if there is no infinite sequence
$t \hookrightarrow t_1 \hookrightarrow t_2 \hookrightarrow \cdots$
starting from $t$.
\emph{Termination} of the relation $\hookrightarrow$ is defined as the absence
of non-terminating terms. In other words, $\hookrightarrow$ is well-founded if and only if it is terminating.
A TRS $\RR$ is \emph{terminating} if $\to_\RR$ is.
For example, the TRS consisting of $((\m{S} \cdot x) \cdot y) \cdot z \to (x \cdot z) \cdot (y \cdot z)$ and 
$(\m{K} \cdot x) \cdot y \to x$ is non-terminating, see for example~\cite{Bar84}.
In contrast, the TRS consisting of only $(\m{K} \cdot x) \cdot y \to x$ is terminating.

Termination of TRSs is often shown by orders.
We say that a rewrite relation is a
\emph{rewrite preorder} or \emph{reduction order} if
it is a preorder or a well-founded order, respectively.

\begin{prop}
\label{P:reduction-order-termination}
A TRS $\RR$ is terminating if $\RR \subseteq {>}$ for some reduction order $>$.
\end{prop}

A powerful method for automated termination analysis is 
the dependency pair framework (to be explained in \autoref{SS:dp}).
Rather than reduction orders, this method 
uses \emph{reduction pairs},
namely  well-founded stable order pairs $({\geqslant}, {>})$ with $\geqslant$ closed under contexts.
A reduction pair $({\geqslant}, {>})$ is \emph{monotone}
if $>$ is also closed under contexts.
We note that $>$ is a reduction order for a monotone reduction pair $({\geqslant}, {>})$.
Conversely, given a reduction order $>$,
the pair $({\geqslant}, {>})$ of the reflexive closure $\geqslant$ and $>$ is a monotone reduction pair.

\subsection{Interpretation-Based Orders}
\label{SS:interpretation}

Algebras (or interpretations) are useful for constructing reduction orders/pairs.
Let $\FF$ be a set of function symbols.
An \emph{$\FF$-algebra} is a pair $(A, \{f_\AA\}_{f \in \FF})$,
where $A$ is a set called a \emph{carrier}, and 
$f_\AA$ is an $n$-ary function on $A$ (called an \emph{interpretation function})
associated with each $f^{(n)} \in \FF$.
Let $\AA = (A, \{f_\AA\}_{f \in \FF})$ be an algebra.
A mapping from $\VV$ to $A$ is called an \emph{assignment} for $\AA$.  The
interpretation $[\alpha]_\AA(t)$ of a term $t$ under an assignment $\alpha$
is inductively defined as follows:
$[\alpha]_\AA(t) = \alpha(t)$ if $t$ is a variable, and 
$[\alpha]_\AA(t) = f_\AA([\alpha]_\AA(t_1),\ldots,[\alpha]_\AA(t_n))$
if $t = f(\seq{t})$.

\begin{exa}
\label{E:algebra}
Let $\FF = \{ \cdot^{(2)}, \m{S}^{(0)}, \m{K}^{(0)} \}$ and 
consider the $\FF$-algebra $\AA$
whose carrier is $\NN$ and 
interpretation functions are given by 
$x \cdot_\AA y = x + y$ and
$\m{S}_\AA = \m{K}_\AA = 1$.
For the assignment $\alpha$ 
with $\alpha(x) = \alpha(y) = 1$
we have $[\alpha]_\AA((\m{K} \cdot x) \cdot y) = \alpha(x) + \alpha(y) + 1 = 3$.
\end{exa}

In the rest of this article,
we assume that the carrier $A$ of an algebra $\AA$ is always endowed with an order pair $({\geqslant},{>})$.
We write $s >_\AA t$ if
$[\alpha]_\AA(s) > [\alpha]_\AA(t)$ for all assignments $\alpha$.
The relation $>_\AA$ is a stable strict order. 
Similarly we write $s \geqslant_\AA t$  if $[\alpha]_\AA(s) \geqslant [\alpha]_\AA(t)$
holds for all assignments $\alpha$.
The relation $\geqslant_\AA$ is a stable quasi-order, and satisfies compatibility with $>_\AA$, namely
${\geqslant_\AA \cdot >_\AA} \subseteq {>_\AA}$ and ${>_\AA \cdot \geqslant_\AA} \subseteq {>_\AA}$.
The $i$-th argument position of $f^{(n)} \in \FF$ with $1 \leqslant i \leqslant n$ is said to be
\begin{itemize}
\item
\emph{strictly simple}  if
$f_\AA(a_1, \ldots, a_i, \ldots, a_n) > a_i$ for all $\seq{a} \in A$;
\item
\emph{weakly simple} if
$f_\AA(a_1, \ldots, a_i, \ldots, a_n) \geqslant a_i$ for all $\seq{a} \in A$;
\item
\emph{strictly monotone} if
\(
f_\AA(a_1,\ldots,a_i,\ldots,a_n) >
f_\AA(a_1,\ldots,b,\ldots,a_n)
\)
follows from $a_i > b$
for all $\seq[n]{a}, b \in A$; and
\item \emph{weakly monotone} if
\(
f_\AA(a_1,\ldots,a_i,\ldots,a_n) \geqslant
f_\AA(a_1,\ldots,b,\ldots,a_n)
\)
follows from $a_i \geqslant b$
for all $\seq[n]{a}, b \in A$.
\end{itemize}
The terminology above carries over to subsets $\GG \subseteq \FF$:
for example, we say that $\AA$ is strictly simple for $\GG$ if every argument position of
$f \in \GG$ is so.
Moreover, if $\GG = \FF$, we may simply call  $\AA$ strictly simple.
The algebra $\AA$ is \emph{well-founded} if $>$ is well-founded;
\emph{normal} if $> {\subseteq} \geqslant$; and
\emph{trivial} if the carrier $A$ is a singleton set.

\begin{exa}[continued from \autoref{E:algebra}]
\label{E:algebra2}
We endow $\AA$ with the order pair $({\geqslant_\NN}, {>_\NN})$.
Then $\AA$ is weakly simple and monotone, and strictly monotone,
but it is not strictly simple, since $x \cdot_\AA y = x$ when $y = 0$. 
Moreover, $\AA$ is well-founded and normal, but not trivial.
\end{exa}

It is known that if $>$ is well-founded, so is $>_\AA$. Similarly,
if $\AA$ is normal, so is $({\geqslant_\AA}, {>_\AA})$.
If $\AA$ is trivial, then $({\geqslant_\AA}, {>_\AA})$ degenerates to $({\TT \times \TT}, \emptyset)$.
If $\AA$ is a weakly monotone algebra, $\geqslant_\AA$ is a
rewrite preorder.
If $\AA$ is a strictly monotone algebra, $>_\AA$ is a reduction order.
If $\AA$ is strictly (resp.\ weakly) simple,
$>_\AA$ (resp. $\geqslant_\AA$) has the subterm property.
An important consequence is:

\begin{propC}[\cite{Zan94}]
\label{P:algebra-redpair}
If $\AA$ is a weakly monotone and well-founded algebra, then $({\geqslant_\AA}, {>_\AA})$ is a reduction pair.
If in addition $\AA$ is strictly monotone, then $({\geqslant_\AA}, {>_\AA})$ is a monotone reduction pair.
\end{propC}

\begin{exa}[continued from \autoref{E:algebra2}]
Consider the TRS $\RR = \{ (\m{K} \cdot x) \cdot y \to x \}$.
From \autoref{P:algebra-redpair}, the pair $({\geqslant_\AA}, {>_\AA})$ is a monotone reduction pair,
and therefore $>_\AA$ is a reduction order.
Furthermore, we  have $\RR \subseteq {>_\AA}$, as
all assignments $\alpha$ satisfy the inequality
\(
[\alpha]_\AA((\m{K} \cdot x) \cdot y) = \alpha(x) + \alpha(y) + 1 > \alpha(x) = [\alpha]_\AA(x)
\).
Hence the termination of $\RR$ is concluded by \autoref{P:reduction-order-termination}.
\end{exa}

The algebra used above is an
instance of the linear polynomial interpretation~\cite{Lan75}.
Formally, a \emph{linear polynomial interpretation} $\AA$ is an algebra 
such that the carrier is $\NN$ ordered by $({\geqslant_\NN}, {>_\NN})$ and
for each $f^{(n)} \in \FF$ its interpretation $f_\AA(x_1, \ldots, x_n)$ is of the form
$c_0 + c_1 x_1 + \cdots + c_n x_n$ with $c_0, c_1, \ldots, c_n \in \NN$.
Such an interpretation $\AA$ is well-founded, normal and also weakly monotone.
In general, the $i$-th argument of a function symbol $f$ is strictly monotone
if and only if $c_i > 0$.
Strict simplicity and weak simplicity are characterized in the same way.

A \emph{max/plus interpretation} $\AA$
associates a combination of $+$ and $\max$ for each function symbol.
{For simplicity, we only consider $\NN$ as the carrier,
and the interpretations are restricted to the following form (adapted from~\cite[Definition~7]{YKS15}):}
For each $f^{(n)} \in \FF$ its interpretation is of the form
$f_\AA(x_1, \ldots, x_n) = \max\{d_0, c_1 x_1 + d_1, \ldots, c_n x_n + d_n\}$
where $d_0$ is a natural number, $d_1, \ldots, d_n \in \ZZ$ and $c_1, \ldots, c_n  \in \{ 0, 1 \}$.
(See \cite{FGM+08} for a more general version.)
Every max/plus interpretation $\AA$ is weakly monotone, well-founded and normal.
The $i$-th argument of a function symbol $f$ is weakly simple 
if and only if $d_i \geqslant 0$ and $c_i > 0$. 
Consider $f_\AA(x_1, x_2) = \max \{1, x_1, x_2 - 1 \}$.
The first argument of $f_\AA$ is weakly simple but not strictly
simple, and the second is not even weakly simple. In fact,
$f_\AA(1, 0) = 1 \ngtr 1$ and $f_\AA(0, 2) = 1 \ngeqslant 2$.

\section{Reduction Orders}
\label{S:reduction-order}

This section studies the WPO, the MSPO, and their natural generalization,
the GWPO, in the reduction-order setting.
Through our analysis of the GWPO, we identify both a ground-total subclass of the MSPO
and a subclass that admits efficient term comparison.

More precisely, we first introduce a variant of (M)SPOs suited to our purposes in \autoref{SS:spo},
and then WPOs and GWPOs in \autoref{SS:GWPO}.
In \autoref{SS:simulation}, we establish a simulation result showing that
the GWPO and the MSPO coincide as classes of reduction orders.
We then discuss two subclasses of GWPOs, and hence of MSPOs:
an efficient subclass in \autoref{SS:gwpo-as-optimized-spo}
and a ground-total subclass in \autoref{SS:ground-totality}.

\subsection{Semantic Path Order}
\label{SS:spo}

Among the many variants of the semantic path order,
we use one based on order pairs and lexicographic extension,
as suggested in~\cite[Section~4.5]{Bor03}.
We employ order pairs to allow interpretations to serve as basic ingredients of the SPO,
and lexicographic extension to relate the SPO to the WPO, which is usually defined using lexicographic extension~\cite{YKS15}.
These choices are explained in due course.

\begin{defi}
\label{D:spo}
Let $({\sqsupseteq}, {\sqsupset})$ be a pair of relations on terms.\footnote{%
Observe that \autoref{D:spo} uses the order pair only when $s$ and $t$ are
headed by function symbols,
so behavior of the orders on variables does not matter.}
The \emph{semantic path order} $>_\spo$ is defined on terms as follows:
$s >_\spo t$ if $s = f(\seq[m]{s})$ and one of the following conditions holds.
\begin{enumerate}[a.]
    \item \label{spo:a}
    $s_i \geqslant_\spo t$ for some $1 \leqslant i \leqslant m$. 
    \item \label{spo:b}
    $t = g(\seq{t})$ and $s >_\spo t_j$ for all $1 \leqslant j \leqslant n$,
    and moreover
    \begin{enumerate}[(i)]
    \item \label{spo:b(i)}
    $s \sqsupset t$, or
    \item \label{spo:b(ii)}
    $s \sqsupseteq t$ and $(\seq[m]{s}) >_\spo^{\lex} (\seq{t})$.
    \end{enumerate}
\end{enumerate}
Here $\geqslant_\spo$ denotes the reflexive closure of $>_\spo$.
\end{defi}

When the signature is infinite,
an SPO induced by an order pair $({\sqsupseteq}, {\sqsupset})$ is not always well-founded,
even if $\sqsupset$ is well-founded.
This is already remarked in \cite[Section~4.5]{Bor03};
see also \cite[Section~3]{Toy04} and \cite[Section~3]{ST13} for related discussions.
Well-foundedness of $>_\spo$ is restored by assuming existence of
an upper bound of arities used in the signature
(cf.~\autoref{P:lex-ext}).
We refer to this condition as \emph{boundedness} of the signature.
Needless to say, a signature is bounded whenever it is finite.

\begin{asm}
The signature is assumed to be bounded in the rest of the paper.
\end{asm}

\begin{thm}
\label{T:spo}
If $({\sqsupseteq}, {\sqsupset})$ is a stable and well-founded order pair on terms,
the induced semantic path order is a stable well-founded order with the subterm property.
\end{thm}

We are not aware of a published proof of the theorem above.
So, for completeness, a proof of \autoref{T:spo} is included in \autoref{S:omitted}.
In general, semantic path orders are not closed under contexts, even when $\sqsupseteq$ is monotone.

\begin{exa}
\label{E:spo-not-monotone}
Let $\AA$ be the linear polynomial interpretation over $\NN$ with $\m{f}_\AA(x) = x + 1$ and $\m{g}_\AA(x) = 0$,
and $>_\spo$ the SPO induced by the reduction pair $({\geqslant_\AA}, {>_\AA})$.
Then we have $\m{g}(x) >_\spo x$ by case~\SPO{a}.
However, $\m{f}(\m{g}(x)) >_\spo \m{f}(x)$ does not hold:
case~\SPO{a} is not applicable, and
case~\SPO{b} cannot be applied either, since $\m{f}(\m{g}(x)) \geqslant_\AA \m{f}(x)$ does not hold.
\end{exa}

For a remedy, Borralleras et al.~\cite{BFR00,Bor03} have proposed to restrict the SPO $>_\spo$
by taking intersection with another preorder $\geqslant$ satisfying the \emph{harmony} property~\cite{Der04}
(originally called quasi-monotonicity~\cite{BFR00}).
Here, a relation $\hookrightarrow$ on terms is \emph{harmonious} with respect to another relation $\leadsto$ on terms if
\[
s_i \leadsto t \quad \Longrightarrow \quad f(s_1, \ldots, s_i, \ldots, s_n) \hookrightarrow f(s_1, \ldots, t, \ldots, s_n)
\]
for all $n$-ary function symbols $f$, integers $1 \leqslant i \leqslant n$, and terms $s_1, \ldots, s_n, t$.
It is not difficult to see that $\hookrightarrow$ is closed under contexts
if and only if $\hookrightarrow$ is harmonious with respect to itself.

\begin{defiC}[(\textnormal{\cite[Definition~4.1.20]{Bor03}})]
A triple $({\geqslant},{\sqsupseteq}, {\sqsupset})$  is a \emph{reduction triple}\footnote{%
Borralleras calls reduction triples reduction \emph{triplets}~\cite{Bor03}.
In the context of dependency pairs, there is another usage of the term in a different meaning~\cite{HM07}.
}
if $\geqslant$ is a rewrite preorder on terms,
$({\sqsupseteq},{\sqsupset})$ is a stable order pair on $\TT$ with $\sqsupset$ well-founded, and
$\sqsupseteq$ is harmonious with respect to $\geqslant$.
\end{defiC}

\begin{defi}
Let $({\geqslant},{\sqsupseteq}, {\sqsupset})$ be a reduction triple and 
$>_\spo$ the semantic path order induced from $({\sqsupseteq}, {\sqsupset})$.
The \emph{monotonic semantic path order} $s >_\mspo t$ is defined as
$s \geqslant t$ and $s >_\spo t$.
\end{defi}

\begin{thm}
\label{T:mspo}
Every monotonic semantic path order is a reduction order.
\end{thm}
Although the proof due to Borralleras et al.~\cite[Theorem~2]{BFR00} goes through,
we present it here to explain the role of $\geqslant$ and the harmony property.
\begin{proof}
Since we have \autoref{T:spo}, the only interesting part is the proof that $>_\mspo$ is monotone.
Let $>_\mspo$ be the MSPO induced by a reduction triple $({\geqslant},{\sqsupseteq}, {\sqsupset})$,
and let $s >_\mspo t$ and $C = f(\bar{u}, \hole, \bar{v})$.
By definition $s \geqslant t$ and $s >_\spo t$.
From the monotonicity of $\geqslant$ we have  $C[s] \geqslant C[t]$, 
and from the harmony property we have $C[s] \sqsupseteq C[t]$.
So, $C[s] >_\spo C[t]$ follows by case~\SPO{b(ii)}, using the subterm property of $>_\spo$.
Hence $C[s] >_\mspo C[t]$.
\end{proof}

\begin{rem}
\label{R:incremental}
The MSPO is incremental in the following sense:
Let $({\geqslant_A}, {\sqsupseteq_A}, {\sqsupset_A})$ and 
$({\geqslant_B}, {\sqsupseteq_B}, {\sqsupset_B})$ be reduction triples.
Suppose that $A$ extends $B$, i.e., 
$\geqslant_A$ extends $\geqslant_B$ and also $({\sqsupseteq_A}, {\sqsupset_A})$ extends $({\sqsupseteq_B}, {\sqsupset_B})$.
Then, the MSPO $>_\mspo^{A}$ induced by $A$ extends the MSPO $>_\mspo^{B}$ induced by $B$.
Since the greater relation is the more profitable for termination proving in general (see \autoref{P:reduction-order-termination}),
we have no reason to use $B$ instead of $A$.
\end{rem}

In the sequel,
we recall constructions of reduction triples~\cite{BFR00,Bor03}.
A quasi-order on the signature is called a \emph{precedence},
and it naturally induces a reduction triple.
A precedence $\succeq$ is called \emph{well-founded} if its strict part $\succ$ is well-founded.

\begin{propC}[\cite{BFR00}]
\label{P:prec-redtriple}
Let $\succeq$ be a well-founded precedence, and define
$s \sqsupseteq_P t$ by $s = t$ or $s = f(\bar{s})$, $t = g(\bar{t})$ and $f \succeq g$, and
$\sqsupset_P$ by the strict part of $\sqsupseteq_P$.
Then $({\TT \times \TT}, {\sqsupseteq_P}, {\sqsupset_P})$ is a reduction triple.
\end{propC}

Note that the harmony property follows from the fact that
$\sqsupseteq_P$ equates terms with the same root symbol.
\emph{Lexicographic path orders} (LPOs)~\cite{KL80} are special instances of MSPOs
constructed with \autoref{P:prec-redtriple}.

\begin{rem}
\label{R:order-pair-precedence}
Generalizing precedences from quasi-orders to order pairs does not increase termination proving power
when the signature is finite: 
Let $({\succeq}, {\succ})$ be an order pair on the signature.
Then $\succeq'$ defined as $\succeq \cup \succ$ is a quasi-order
by the compatibility of $({\succeq}, {\succ})$.
Moreover, $\succeq'$ extends $\succeq$, and also the strict part $\succ'$ extends $\succ$.
So,  $\succeq'$ is preferred to $({\succeq}, {\succ})$ when using them with the construction of \autoref{P:prec-redtriple},
because the MSPO is incremental, as remarked in \autoref{R:incremental}.
\end{rem}

For a reduction pair $({\geqslant}, {>})$,
the quasi-order $\geqslant$ is 
monotone,
and therefore harmonious with respect to itself.
So, every reduction pair can be turned into a reduction triple as follows.

\begin{propC}[\cite{BFR00}]
\label{P:redpair-redtriple}
If $({\geqslant}, {>})$ is a reduction pair,
$({\geqslant}, {\geqslant}, {>})$ is a reduction triple.
\end{propC}

In particular, since $({\TT \times \TT}, {\emptyset})$ is a reduction pair,
$({\TT \times \TT}, {\TT \times \TT}, {\emptyset})$ is a reduction triple.
We call it the \emph{trivial} reduction triple.

We illustrate the use of \autoref{P:redpair-redtriple} with the MSPO by an example.

\begin{exa}
We show the termination of the TRS $\RR = \{ \m{f}(\m{f}(x)) \to \m{f}(\m{g}(\m{f}(x))) \}$
taken from \cite{Zan94}.
Let $\AA$ be the linear polynomial interpretation over $\NN$ with $\m{f}_\AA(x) = x + 1$ and $\m{g}_\AA(x) = 0$
(the same one as \autoref{E:spo-not-monotone}),
and $>_\mspo$ the MSPO induced by the reduction triple $({\geqslant_\AA}, {\geqslant_\AA}, {>_\AA})$.
The inequality $\m{f}(\m{f}(x)) \geqslant_\AA \m{f}(\m{g}(\m{f}(x)))$
boils down to $x + 2 > 1$ for all $x \in \NN$, which trivially holds.
In the same way, 
we can confirm $\m{f}(\m{f}(x)) >_\AA t$ for all subterms $t$ of $\m{f}(\m{g}(\m{f}(x)))$.
So, $\m{f}(\m{f}(x)) >_\spo \m{f}(\m{g}(\m{f}(x)))$ follows from successive application of case~\SPO{b(i)}.
Thus, we conclude $\RR \subseteq {>_\mspo}$ and thereby the termination.
\end{exa}

Note that the strict order $>$ in \autoref{P:redpair-redtriple} is not the strict part of $\geqslant$ in general.
This essentially explains why we adopt a formulation based on order pairs in \autoref{D:spo}.
The issue is most evident when they are induced by a polynomial interpretation.

\begin{exa}
\label{E:not-strict-part}
Consider the linear polynomial interpretation $\AA$ over $\NN$ with the interpretation $\m{f}_\AA(x) = 2 x$,
and the induced reduction triple $({\geqslant_\AA}, {\geqslant_\AA}, {>_\AA})$.
On the one hand we have $\m{f}(\m{f}(x)) \geqslant_\AA \m{f}(x)$ from
$4x \geqslant 2x$ for all $x \in \NN$,
but not $\m{f}(x) \geqslant_\AA \m{f}(\m{f}(x))$ as 
$2x \geqslant 4x$ does not hold for $x = 1$.
On the other hand $\m{f}(\m{f}(x)) >_\AA \m{f}(x)$ does not hold either,
because $4 x > 2x$ does not hold for $x = 0$.
\end{exa}

It is known that 
\autoref{P:redpair-redtriple}
can be improved by marking root symbols; see~\cite[Proposition~4.2]{BFR00}.
(A similar idea is used by the dependency pair method discussed in \autoref{S:reduction-pair}.)
In this article, we use a special case based on algebras.
We first introduce a distinguished copy $\FF^\sharp$ of the original signature $\FF$,
consisting of a fresh $n$-ary function symbol $f^\sharp$ for each $f^{(n)} \in \FF$.
For an $(\FF \cup \FF^\sharp)$-algebra $\AA$,
we define $s \geqslant_{\AA}^{\sharp} t$ on $\TT(\FF, \VV)$ by
$s = t$ or $s = f(\bar{s})$, $t = g(\bar{t})$ and $f^\sharp(\bar{s}) \geqslant_\AA g^\sharp(\bar{t})$.
Similarly, we define  $s >_{\AA}^{\sharp} t$ on $\TT(\FF, \VV)$ by
$s = f(\bar{s})$, $t = g(\bar{t})$ and $f^\sharp(\bar{s}) >_\AA g^\sharp(\bar{t})$.
The construction below is based on the observation that 
$\geqslant_\AA^\sharp$ is harmonious with respect to $\geqslant_\AA$
if $\AA$ is weakly monotone for $\FF^\sharp$.

\begin{prop}
\label{P:algebra-redtriple}
The triple $({\geqslant_\AA}, {\geqslant_{\AA}^\sharp}, {>_{\AA}^\sharp})$ is a reduction triple on $\TT(\FF, \VV)$
for every weakly monotone and well-founded  $(\FF \cup \FF^\sharp)$-algebra $\AA$. \qed
\end{prop}

There is a variant of the construction above that allows us to use non-weakly-monotone algebras,
at the cost of using the equivalence
relation $=_\AA$ as the first component of the reduction triple~\cite{HM04}.
Here, $s =_\AA t$ is defined on $\TT(\FF, \VV)$ as
$[\alpha]_\AA(s) = [\alpha]_\AA(t)$ for all assignments $\alpha$.

\begin{prop}
\label{P:model-redtriple}
The triple $({=_\AA}, {\geqslant_{\AA}^\sharp}, {>_{\AA}^\sharp})$ is a reduction triple on $\TT(\FF, \VV)$
for every well-founded  $(\FF \cup \FF^\sharp)$-algebra $\AA$.
\end{prop}
\begin{proof}
For the harmony property, let $s =_\AA t$ and $C = f(\bar{u}, \hole, \bar{v})$.
From $s =_\AA t$ we obtain
$f^\sharp(\bar{u}, s, \bar{v}) =_\AA f^\sharp(\bar{u}, t, \bar{v})$,
and thus
$C[s] \geqslant_\AA^\sharp C[t]$.
The other properties are easy to verify.
\end{proof}

\begin{exa}
\label{E:Dershowitz}
Consider the TRS $\RR = \{ \m{f}(\m{a}, \m{b}, x) \to \m{f}(x, x, x) \}$.
We mimic the termination proof of \cite[Example~2]{Zan95} by the finite algebra $\AA$ over $\{ 0, 1 \}$ defined as follows:
\begin{align*}
\m{f}_\AA(x, y, z) & = 0 & \m{a}_\AA &= 1 & \m{b}_\AA &= 0
&
\m{f}^\sharp_\AA(x, y, z) & = \begin{cases*}
1 & \text{if $x \neq y$} \\
0 & \text{otherwise}
\end{cases*}
&
\m{a}^\sharp_\AA &= 1 & \m{b}^\sharp_\AA &= 0
\end{align*}
We consider the MSPO $>_\mspo$ induced by the reduction triple $({=_\AA}, {\geqslant_{\AA}^\sharp}, {>_{\AA}^\sharp})$.
Obviously, we have $\m{f}(\m{a}, \m{b}, x) =_\AA \m{f}(x, x, x)$, and also $\m{f}^\sharp(\m{a}, \m{b}, x) >_\AA \m{f}^\sharp(x, x, x)$.
So $\m{f}(\m{a}, \m{b}, x) >_\spo \m{f}(x, x, x)$ follows from case~\SPO{b(i)}.
Finally, we conclude $\RR \subseteq {>_\mspo}$ and thereby termination.
\end{exa}

\subsection{Generalizing WPO}
\label{SS:GWPO}

The \emph{weighted path order}
(WPO) is a class of reduction orders introduced by Yamada, Kusakari and Sakabe~\cite{YKS15}.
The definition of WPOs is based on an algebra $\AA$ and
a precedence $\succeq$.\footnote{%
{In contrast to \cite[Definition~3]{YKS15}, we do not consider the status extension 
that allows permutation of arguments.
This is merely to simplify the presentation.}
}
A WPO compares terms $s$ and $t$ as a Knuth--Bendix order (KBO)~\cite{KB70} does:
If $s >_\AA t$, we can immediately conclude that $s$ is strictly greater than $t$.
If the weak inequality $s \geqslant_\AA t$ holds
then we can compare an argument of $s$ with $t$, or
$f = \rt(s)$ and $g = \rt(t)$ by the precedence $\succeq$.
When the weak inequality $f \succeq g$ holds, then the arguments of $s$ and $t$ can be compared lexicographically.

\begin{defiC}[\cite{YKS15}]
\label{D:wpo-reduction-order}
Let $\AA$ be an ordered $\FF$-algebra and $\succeq$ a precedence.
The \emph{weighted path order} $>_\wpo$ is defined
on terms over $\FF$ as follows: $s >_\wpo t$ if
\begin{enumerate}[1.]
\item \label{wpo:1} $s >_\AA t$, or
\item \label{wpo:2} $s \geqslant_\AA t$, $s = f(\seq[m]{s})$, and one of the following
conditions holds.
\begin{enumerate}[a.]
\item \label{wpo:2a}
$s_i \geqslant_\wpo t$ for some $1 \leqslant i \leqslant m$.
\item \label{wpo:2b} \smallskip
$t = g(\seq{t})$ and $s >_\wpo t_j$ for all $1 \leqslant j \leqslant n$, and
moreover
\smallskip
    \begin{enumerate}[(i)]
    \item \label{wpo:2b(i)}
    $f \succ g$, or
    \item \label{wpo:2b(ii)}
    $f \succeq g$ and $(\seq[m]{s}) >_\wpo^\lex (\seq{t})$.
    \end{enumerate}
\end{enumerate}
\end{enumerate}
Here $\geqslant_\wpo$ denotes the reflexive closure of $>_\wpo$.
\end{defiC}

We note that generalizing precedences to order pairs does not increase termination proving power of the WPO;
see~\autoref{R:order-pair-precedence}.

\begin{thmC}[\cite{YKS15}]
\label{T:wpo-reduction-order}
For every well-founded, normal, weakly simple and monotone algebra and well-founded precedence,
the induced relation $>_\wpo$ is a reduction order (with the subterm property).
\end{thmC}

We note that the normality requirement was identified by \cite{TSSY20}.
Weak simplicity of the algebra is essential for well-foundedness of the WPO in \autoref{T:wpo-reduction-order}.
In fact, otherwise the order may not be irreflexive.

\begin{exa}
\label{E:wpo-weak-simple}
Let $\AA$ be the weakly monotone and normal algebra on $\NN$ given by $\m{a}_\AA = 1$ and $\m{f}_\AA(x) = 0$.
The WPO $>_\wpo$ induced by $\AA$ and any precedence satisfies $\m{f}(\m{a}) >_\wpo \m{f}(\m{a})$
which can be seen by applying case~\WPO{2a} and then case~\WPO{1}.
\end{exa}

In general, it is known that the WPO generalizes the LPO and the KBO in the following sense:
KBOs are WPOs induced by a restricted form of linear polynomial interpretations over $\NN$, and 
LPOs are WPOs induced by a restricted form of max/plus interpretations~\cite{YKS15}.
So every termination proof by an LPO or KBO can be simulated by a corresponding WPO.

We generalize the WPO in the spirit of Kamin and L{\'e}vy.
Their idea to obtain the SPO is to generalize precedence of the LPO
to order pairs on terms. Following this, we generalize the WPO so as to take 
two order pairs on terms:
one takes the role of algebra, and the other that of precedence.

\begin{defi}
\label{D:gwpo}
Let $({\sqsupseteq_A}, {\sqsupset_A})$ and $({\sqsupseteq_B}, {\sqsupset_B})$ be pairs of relations on terms.
The \emph{generalized weighted path order} $>_\gwpo$ (GWPO) is defined as follows: 
$s >_\gwpo t$ if
\begin{enumerate}[1.]
  \item \label{gwpo:1} $s \sqsupset_A t$, or
  \item \label{gwpo:2} $s = f(\seq[m]{s}) \sqsupseteq_A t$ and one of the following conditions holds:
  \begin{enumerate}[a.]
    \item \label{gwpo:2a} $s_i \geqslant_\gwpo t$ for some argument $s_i$.
    \item \label{gwpo:2b} $t = g(\seq[n]{t})$, $s >_\gwpo t_j$ for all arguments $t_j$, and moreover
    \begin{enumerate}[(i)]
      \item \label{gwpo:2b(i)} $f(\seq[m]{s}) \sqsupset_B g(\seq[n]{t})$, or
      \item \label{gwpo:2b(ii)} $f(\seq[m]{s}) \sqsupseteq_B g(\seq[n]{t})$ and $(\seq[m]{s}) >_\gwpo^\lex (\seq[n]{t})$.
    \end{enumerate}
  \end{enumerate}
\end{enumerate}
\end{defi}

It is a nice exercise to compare \autoref{D:gwpo} and \autoref{D:spo}:
an SPO induced by a pair $({\sqsupseteq}, {\sqsupset})$ is identical to the GWPO 
induced by  $({\TT \times \TT}, {\emptyset})$ and $({\sqsupseteq}, {\sqsupset})$.
It is also instructive to compare \autoref{D:gwpo} and \autoref{D:wpo-reduction-order}:
The order pair $({\sqsupseteq_A}, {\sqsupset_A})$ corresponds to $({\geqslant_{\AA}}, {>_{\AA}})$ induced by an algebra $\AA$,
and  $({\sqsupseteq_B}, {\sqsupset_B})$ corresponds to comparison by a precedence.
Interestingly, the version of WPO formalized in Isabelle/HOL (\cite[Section~3]{TSSY20}) is based on a similar idea,
but does not generalize precedences to order pairs on terms.

Not surprisingly, a sufficient condition for a GWPO to be a stable and well-founded order
is a merger of that for an SPO (\autoref{T:spo}) and that for a WPO (\autoref{T:wpo-reduction-order}).
We say that an order pair $({\sqsupseteq_A}, {\sqsupset_A})$  on terms has the subterm property
if $\prsuperterm \subseteq {(\sqsupseteq_A \cup \sqsupset_A)}$.
Hereafter, for brevity we may simply write the subscript $A$ for a pair $({\sqsupseteq_A}, {\sqsupset_A})$
or a triple $({\geqslant_A}, {\sqsupseteq_A}, {\sqsupset_A})$,
if it does not cause any confusion.

\begin{thm}
\label{T:gwpo}
If $({\sqsupseteq_A}, {\sqsupset_A})$ and $({\sqsupseteq_B}, {\sqsupset_B})$ are well-founded stable order pairs and
the order pair $A$ has the subterm property,
the induced relation $>_\gwpo$ is a stable and well-founded order with the subterm property.
\end{thm}

In the next section we show the theorem using a simulation result.
The subterm property of $A$ is necessary for
well-foundedness. The necessity is witnessed by the WPO in
\autoref{E:wpo-weak-simple}, as it can be regarded as a GWPO induced by $({\geqslant_\AA},{>_\AA})$.
By definition, the subterm property of $\sqsupseteq_A$ is a sufficient condition for that of the order pair $(\sqsupseteq_A, \sqsupset_A)$.
Here, we allow the weaker condition $\prsuperterm \subseteq {(\sqsupseteq_A \cup \sqsupset_A)}$
for convenience in \autoref{SS:gwpo-as-optimized-spo}
(so that \autoref{C:kbo-like-spo} can be derived using \autoref{T:gwpo-spo}).

As the SPO has the monotonic variant MSPO,
so does the GWPO\@.
We call it the \emph{monotonic generalized weighted path order},
or MGWPO for short.

\begin{defi}
\label{D:mgwpo}
Let $({\geqslant_A}, {\sqsupseteq_A}, {\sqsupset_A})$ and $({\geqslant_B}, {\sqsupseteq_B}, {\sqsupset_B})$ be triples of relations on terms.
The \emph{monotonic generalized weighted path order} $>_\mgwpo$ is the intersection
${\geqslant_A} \cap {\geqslant_B} \cap {>_\gwpo}$ where $>_\gwpo$ is induced by  $({\sqsupseteq_A}, {\sqsupset_A})$ and $({\sqsupseteq_B}, {\sqsupset_B})$.
\end{defi}

Now, we arrive at the main theorem of this section.

\begin{thm}
\label{T:mgwpo}
The relation $>_\mgwpo$ is a reduction order
for reduction triples $({\geqslant_A}, {\sqsupseteq_A}, {\sqsupset_A})$ and $({\geqslant_B}, {\sqsupseteq_B}, {\sqsupset_B})$
with $(\sqsupseteq_A, \sqsupset_A)$ having the subterm property. 
\end{thm}

The theorem is also a consequence of the simulation result
in the next section.
Next,  we confirm that our MGWPO indeed generalizes the WPO.

\begin{prop}
Let $>_{\wpo}$ be the WPO induced by a normal algebra $\AA$ and a precedence $\succeq$.
Then $>_{\wpo}$ is identical to $>_\mgwpo$
induced by $({\geqslant_\AA}, {\geqslant_\AA}, {>_\AA})$ and $P$,
where $P$ is constructed from $\succeq$ as in \autoref{P:prec-redtriple}.
\end{prop}
\begin{proof}
It is easy to see that $>_{\wpo}$ is identical to $>_\gwpo$, and therefore ${>_\mgwpo} \subseteq {>_\wpo}$.
Moreover, by the normality of $\AA$, one can show that ${>_\wpo} \subseteq {\geqslant_\AA}$.
This entails the converse ${>_\wpo} \subseteq {>_\mgwpo}$.
\end{proof}

Using the same construction, we can see that \autoref{T:wpo-reduction-order} follows from \autoref{T:mgwpo}:
When $\AA$ is weakly simple, we have the subterm property of $\geqslant_\AA$ and thereby that of $({\geqslant_\AA}, {>_\AA})$.

Also, we can see that our MGWPO indeed generalizes MSPO.

\begin{prop}
\label{P:mspo-mgwpo}
Let $>_{\mspo}$ be the MSPO induced by a reduction triple $({\geqslant}, {\sqsupseteq}, {\sqsupset})$.
Then $>_{\mspo}$ is identical to $>_\mgwpo$
induced by the reduction triples $({\TT \times \TT}, {\TT \times \TT}, {\emptyset})$ and 
$({\geqslant}, {\sqsupseteq}, {\sqsupset})$. \qed
\end{prop}

Now, we illustrate the use of GWPOs by examples.

\begin{exa}
\label{E:SK}
Consider the rewrite system for combinatory logic:
\begin{align*}
((\m{S} \cdot x) \cdot y) \cdot z & \to (x \cdot z) \cdot (y \cdot z)
&
(\m{K} \cdot x) \cdot y & \to x
\end{align*}
While the system is well known to be non-terminating (see for example~\cite{Bar84}),
the following system $\RR$ obtained by flipping the $\m{S}$ rule is actually terminating.\footnote{%
The observation
that the rule $(x \cdot z) \cdot (y \cdot z) \to ((\m{S} \cdot x) \cdot y) \cdot z$ is terminating dates back to \cite{BK79}.}
\begin{align*}
(x \cdot z) \cdot (y \cdot z) & \to ((\m{S} \cdot x) \cdot y) \cdot z
&
(\m{K} \cdot x) \cdot y & \to x
\end{align*}
To see it,
let $\AA$ be the linear polynomial interpretation over $\NN$
defined as follows:
\begin{align*}
x \cdot_\AA y &= x + y + 1 & \m{S}_\AA &= 0 & \m{K}_\AA &= 0
&
x \cdot^\sharp_\AA y &= y & \m{S}^\sharp_\AA &= 0 & \m{K}^\sharp_\AA &= 0
\end{align*}
Intuitively, $\AA$ captures the decrease of the number of occurrences of
$\cdot$ in the second rule with $x \cdot_\AA y = x + y + 1$, and the
decrease of the depth of the rightmost $z$ in the first rule
with $x \cdot^\sharp_\AA y = y$.
The monotone generalized weighted path order $>_\mgwpo$ induced by
the reduction triples $({\geqslant_\AA}, {\geqslant_\AA}, {>_\AA})$ and
$({\geqslant_\AA}, {\geqslant_\AA^\sharp}, {>_\AA^\sharp})$ satisfies 
$(x \cdot z) \cdot (y \cdot z) >_\mgwpo ((\m{S} \cdot x) \cdot y) \cdot z$.
The comparison proceeds as follows: First,
$(x \cdot z) \cdot (y \cdot z) \geqslant_\AA ((\m{S} \cdot x) \cdot y) \cdot z$ holds,
as it boils down to the valid inequality $x + y + 2z + 3 \geqslant x + y + z + 3$.
So it remains  to confirm a decrease in $>_\gwpo$. Let $t = (x \cdot z) \cdot (y \cdot z)$.
Then we have the following derivation.
\begin{center}
\begin{prooftree}
  \hypo{(x \cdot z) \cdot^\sharp (y \cdot z) >_\AA ((\m{S} \cdot x) \cdot y) \cdot^\sharp z}
    \hypo{t >_\AA (\m{S} \cdot x) \cdot y}
  \infer{1}[\GWPO{1}]{t >_\gwpo (\m{S} \cdot x) \cdot y}
    \hypo{t >_\AA z}
  \infer{1}[\GWPO{1}]{t >_\gwpo z}
\infer{3}[\GWPO{2b(i)}]{t >_\gwpo ((\m{S} \cdot x) \cdot y) \cdot z}
\end{prooftree}
\end{center}
Here
$(x \cdot z) \cdot^\sharp (y \cdot z) >_\AA ((\m{S} \cdot x) \cdot y) \cdot^\sharp z$ 
follows from the validity of $y + z + 1 > z$.
In a similar way, $(\m{K} \cdot x) \cdot y >_\mgwpo x$ is verified.  As
$\RR \subseteq {>_\mgwpo}$ holds, the termination of $\RR$ is concluded.
\end{exa}

The termination of $\RR$ above cannot be shown by the original weighted path order
if algebras are restricted to be linear polynomial interpretations over $\NN$,
see~\autoref{P:S-WPO}.

\begin{exa}
\label{E:cge2}
Consider the TRS $\RR$ for groups with two commuting
endomorphisms~\cite{SL06}:
\begin{align*}
(- x) + x & \xto{1} \m{0}
&
(x + y) + z & \xto{3} x + (y + z) 
&
\m{g}(x) + \m{g}(y) & \xto{5} \m{g}(x + y)
\\
\m{0} + x & \xto{2} x
&
\m{f}(x) + \m{f}(y) & \xto{4} \m{f}(x + y)
&
\m{f}(x) + \m{g}(y) & \xto{6} \m{g}(y) + \m{f}(x)
\end{align*}
A difficulty is rule 6, which is
a terminating instance of the non-terminating rule $x + y \to y + x$.
We consider the linear polynomial interpretation $\AA$ 
over $\NN$  with
\begin{align*}
\m{0}_{\AA} &= 1
&
\m{f}_{\AA}(x) &= x + 1
&
\m{g}_{\AA}(x) &= x + 1
&
-_{\AA}(x) &= x + 1
&
x +_{\AA} y &= x + y + 1
\end{align*}
and the linear polynomial interpretation $\BB$ over $\NN$ for the extended signature $\FF \cup \FF^\sharp$:
\begin{align*}
\m{0}_{\BB} &= 0
&
\m{f}_{\BB}(x) &= 1
&
\m{g}_{\BB}(x) &= 0
&
-_{\BB}(x) &= x + 1
&
x +_{\BB} y &= x + y + 1
\\
\m{0}^\sharp_{\BB} &= 0
&
\m{f}^\sharp_{\BB}(x) &= 1
&
\m{g}^\sharp_{\BB}(x) &= 0
&
-^\sharp_{\BB}(x) &= x + 1
&
x +^\sharp_{\BB} y &= x
\end{align*}
Roughly, $\AA$ captures the decreases in term size
while $\BB$ captures the change from $\m{f}$ to $\m{g}$ in the first argument of $+$ for rule~6.
This can be seen in $x +^\sharp_{\BB} y = x$, $\m{f}_{\BB}(x) = 1$ and $\m{g}_{\BB}(x) = 0$.
From \autoref{T:mgwpo} the MGWPO $>_\mgwpo$ induced by the reduction triples $({\geqslant_\AA}, {\geqslant_\AA}, {>_\AA})$ and $({\geqslant_\BB}, {\geqslant_\BB^\sharp}, {>_\BB^\sharp})$
is a reduction order.
To conclude the termination of $\RR$ with \autoref{P:reduction-order-termination},
we verify $\RR \subseteq {>_\mgwpo}$.
Checking $\RR \subseteq {\geqslant_\AA}$ and $\RR \subseteq {\geqslant_\BB}$ is routine work.
So, it remains to verify $\RR \subseteq {>_\gwpo}$.
For rule 6,
we use case~\GWPO{2b(i)} with
\(
\m{f}(x) +^\sharp \m{g}(y) >_\BB \m{g}(y) +^\sharp \m{f}(x)
\)
and then case~\GWPO{1} for recursive comparison.
For rule 3 we 
similarly use case~\GWPO{2b(i)} with $(x + y) +^\sharp z >_\BB x +^\sharp (y + z)$
and then case~\GWPO{1}.
The other rules are oriented by $>_\AA$ and therefore by case~\GWPO{1}.
\end{exa}

\begin{rem}
\label{R:algebra-redtriple}
We cannot naively use \autoref{P:algebra-redtriple} for 
the first reduction triple $({\geqslant_A}, {\sqsupseteq_A}, {\sqsupset_A})$ of an MGWPO,
since \autoref{T:mgwpo} demands the subterm property of $(\sqsupseteq_A, \sqsupset_A)$.
For example, $\m{f}(\m{a}) \geqslant_\AA^\sharp \m{a}$ does not hold for a weakly simple algebra $\AA$ on $\NN$ with
\( \m{a}_\AA = 0 \),
\( \m{a}^\sharp_\AA = 1 \), and
\( \m{f}^\sharp_\AA(x) = x \).
So a typical choice for the first component is
to use $({\geqslant_\AA}, {\geqslant_\AA}, {>_\AA})$ with an algebra $\AA$ on the original signature,
as in the examples above. (We however have more to say about this in the next section.)
Moreover,  if $\AA$ is normal,
then ${>_\gwpo} \subseteq {\geqslant_\AA}$ holds
and therefore  $s >_\mgwpo t$ is equivalent to $s \geqslant_B t$ and $s >_\gwpo t$.
\end{rem}

\subsection{Simulating GWPOs by SPOs}
\label{SS:simulation}

We present a simulation result between (M)GWPOs and (M)SPOs.
This result allows us to establish desired properties of (M)SPOs
via those of (M)GWPOs, and vice versa.

To this end, we recall \emph{lexicographic combination}, which provides a way 
to combine multiple order pairs on the same set $X$ into a more complex one on $X$.

\begin{defi}
\label{D:lex-comb}
Let $({\geqslant_A}, {>_A})$ and $({\geqslant_B}, {>_B})$ be 
pairs of relations on the same set $X$.  The \emph{lexicographic combination}
$({\geqslant_{AB}}, {>_{AB}})$ is the pair of relations on $X$ defined
as follows:
\begin{itemize}
\item
$x \geqslant_{AB} y$ if $x >_A y$, or both $x \geqslant_A y$ and
$x \geqslant_B y$.
\item
$x >_{AB} y$ if $x >_A y$, or both $x \geqslant_A y$ and $x >_B y$.
\end{itemize}    
\end{defi}

Let $({\geqslant_A}, {>_A})$ and $({\geqslant_B}, {>_B})$ be order pairs on the same set.
Then $({\geqslant_{AB}}, {>_{AB}})$ is again an order pair.
Moreover, it preserves  well-foundedness:
If both  $({\geqslant_A}, {>_A})$ and $({\geqslant_B}, {>_B})$ are well-founded,
so is $({\geqslant_{AB}}, {>_{AB}})$.
The same goes for stability.

The simulation result is based on lexicographic combination:
Given a GWPO $>_\gwpo$ induced by two order pairs $({\sqsupseteq_A}, {\sqsupset_A})$ and $({\sqsupseteq_B}, {\sqsupset_B})$,
we consider the SPO $>_\spo$ induced by the lexicographic combination $({\sqsupseteq_{AB}},{\sqsupset_{AB}})$.
Our goal is to show that $>_\gwpo$ and $>_\spo$ are identical.

Before presenting proofs, we illustrate how the derivation of 
$t = (x \cdot z) \cdot (y \cdot z) >_\gwpo ((\m{S} \cdot x) \cdot y) \cdot z$
in \autoref{E:SK} is simulated by the semantic path order.

\begin{exa}[continued from \autoref{E:SK}]
\label{E:spo}
Under the construction above, $({\sqsupseteq_{AB}},{\sqsupset_{AB}})$ is the lexicographic combination of
$({\geqslant_\AA}, {>_\AA})$ and $({\geqslant_\AA^{\sharp}}, {>_\AA^{\sharp}})$.
Since ${>_\AA} \subseteq {\sqsupset_{AB}}$,
we can first apply case~\SPO{b(i)} of $>_\spo$ instead of case~\GWPO{2b(i)} of $>_\gwpo$,
which gives rise to two obligations, namely $(x \cdot z) \cdot (y \cdot z) >_\spo (\m{S} \cdot x) \cdot y$ and
$(x \cdot z) \cdot (y \cdot z) >_\spo z$.
The former is derived by applying case~\SPO{b(i)} and case~\SPO{a} repeatedly:
\begin{center}
\begin{prooftree}
  \hypo{t \sqsupset_{AB} (\m{S} \cdot x) \cdot y}
    \hypo{t \sqsupset_{AB} \m{S} \cdot x}
      \hypo{t \sqsupset_{AB} \m{S}}
    \infer{1}[\SPO{a}]{t >_\spo \m{S}}
        \hypo{x \geqslant_\spo x}
      \infer{1}[\SPO{a}]{x \cdot z \geqslant_\spo x}
    \infer{1}[\SPO{a}]{t >_\spo x}
  \infer{3}[\SPO{b(i)}]{t >_\spo \m{S} \cdot x}
      \hypo{y \geqslant_\spo y}
    \infer{1}[\SPO{a}]{y \cdot z \geqslant_\spo y}
  \infer{1}[\SPO{a}]{t >_\spo y}
\infer{3}[\SPO{b(i)}]{t >_\spo (\m{S} \cdot x) \cdot y}
\end{prooftree}
\end{center}
Notice that, $t >_\AA (\m{S} \cdot x) \cdot y$ and the weak simplicity of $\AA$ guarantees
$t >_\AA u$ and thereby $t \sqsupset_{AB} u$ for all subterms $u$ of $(\m{S} \cdot x) \cdot y$ in the simulating derivation.
Similarly, the latter obligation $(x \cdot z) \cdot (y \cdot z) >_\spo z$ is derived by repeatedly applying case~\SPO{a}.
\end{exa}

As shown in the example, the subterm property of $({\sqsupseteq_A}, {\sqsupset_A})$ allows us to simulate case~\GWPO{1} of $>_\gwpo$
by case~\SPO{b(i)} and case~\SPO{a} of $>_\spo$.

\begin{lem}
\label{L:gwpo-spo}
Suppose that $({\sqsupseteq_A}, {\sqsupset_A})$ is a stable order pair with the subterm property.
If $s >_\gwpo t$ then $s >_\spo t$.
\end{lem}
\begin{proof}
We prove the claim by induction on $|s| + |t|$.
Let $s >_\gwpo t$. We analyze the derivation of $s >_\gwpo t$.
\begin{enumerate}[1.]
\item Assume that $s >_\gwpo t$ is derived from $s \sqsupset_A t$. We distinguish three cases.
\begin{itemize}
\item Suppose that $s$ is a variable. We further analyze whether $t$ contains $s$ or not.
If so, by the subterm property of $A$, we have $s \sqsupset_A t \sqsupseteq_A s$ or 
$s \sqsupset_A t \sqsupset_A s$, which is a contradiction; otherwise, by substituting $t$
into $s$ we obtain $t \sqsupset_A t$, which is again a contradiction.

\item Suppose that $s$ is not a variable while $t$ is a variable.
If $t$ does not occur in $s$ then by substituting $s$ into $t$ we obtain $s \sqsupset_A s$, a contradiction.
So $t$ occurs in $s$. We obtain $s >_\spo t$ by the
subterm property of $>_\spo$.

\item
Suppose that $s$ is not a variable, and $t = g(\seq{t})$.  Since $A$ has the subterm property,
for every $1 \leqslant j \leqslant n$ we have $s \sqsupset_A t \sqsupseteq_A t_j$ or $s \sqsupset_A t \sqsupset_A t_j$,
which leads to $s >_\gwpo t_j$. 
Thus, the induction hypothesis yields $s >_\spo t_j$ for all $j$.
Moreover, $s \sqsupset_{AB} t$ follows from $s \sqsupset_A t$.
Hence, $s >_\spo t$ follows from case~\SPO{b(i)}.
\end{itemize}

\item Assume that $s = f(\seq[m]{s})$ and $s \sqsupseteq_A t$.
\begin{enumerate}[a.]
    \item Suppose that $s >_\gwpo t$ is derived from $s \sqsupseteq_A t$
    and $s_i \geqslant_\gwpo t$ for some $i$.
    Then $s >_\spo t$ is obtained by case~\SPO{a} using the induction hypothesis.
    \item Suppose that $t = g(\seq[n]{t})$ and $s >_\gwpo t_j$ for all $1 \leqslant j \leqslant n$.
    By the induction hypothesis $s >_\spo t_j$ for all $1 \leqslant j \leqslant n$.
    \begin{enumerate}[(i)]
    \item
    If $s >_\gwpo t$ is derived from case~\GWPO{2b(i)},
    then we have $s \sqsupset_{AB} t$. So $s >_\spo t$ is obtained by case~\SPO{b(i)}. 
    \item If $s >_\gwpo t$ is derived from case~\GWPO{2b(ii)},
    then we have $s \sqsupseteq_{AB} t$ and $(\seq[m]{s}) >_\spo^\lex (\seq[n]{t})$ by the induction hypothesis.
    So $s >_\spo t$ is obtained by case~\SPO{b(ii)}. \qedhere 
    \end{enumerate}
\end{enumerate}
\end{enumerate}
\end{proof}

Next we prove the converse direction of \autoref{L:gwpo-spo}.

\begin{lem}
\label{L:spo-gwpo}
Suppose that $({\sqsupseteq_A}, {\sqsupset_A})$ is an order pair with the subterm property.
If $s >_\spo t$ then $s >_\gwpo t$.
\end{lem}
\begin{proof}
We prove the claim by induction on $|s| + |t|$.
Let $s = f(\seq[m]{s})$.
We analyze the derivation of $s >_\spo t$.
\begin{enumerate}[a.]
\item
Suppose that $s >_\spo t$ is derived by case~\SPO{a}.
The induction hypothesis yields $s_i \geqslant_\gwpo t$ for some $i$.
Then, by the definition of the GWPO, $s_i \sqsupseteq_A t$ or $s_i \sqsupset_A t$ holds.
\begin{itemize}
  \item In the former case, using the subterm property of $A$, we obtain $s \sqsupseteq_A s_i \sqsupseteq_A t$
  or $s \sqsupset_A s_i \sqsupseteq_A t$.
  So we have $s >_\gwpo t$ by case~\GWPO{2a} or case~\GWPO{1}. 
  \item In the latter case, using the subterm property of $A$, we obtain $s \sqsupseteq_A s_i \sqsupset_A t$
  or $s \sqsupset_A s_i \sqsupset_A t$. Therefore $s >_\gwpo t$ by case~\GWPO{1}.
\end{itemize}
\item
Suppose that $s >_\spo t$ is derived by case~\SPO{b}.
Let $t = g(\seq[n]{t})$ and assume $s >_\spo t_j$ for all $1 \leqslant j \leqslant n$.
By the induction hypothesis we have  $s >_\gwpo t_j$ for all $1 \leqslant j \leqslant n$.
\begin{enumerate}[(i)]
    \item If $s \sqsupset_{AB} t$, we further analyze how it follows:
    If $s \sqsupset_A t$ then we have $s >_\gwpo t$ by case~\GWPO{1}.
    Otherwise, we conclude with case~\GWPO{2b(i)}.
    \item If $s \sqsupseteq_{AB} t$, we further analyze how it follows:
    If $s \sqsupset_A t$ then we have $s >_\gwpo t$ by case~\GWPO{1}.
     Otherwise, we conclude with case~\GWPO{2b(ii)} using the induction hypothesis. \qedhere
\end{enumerate}
\end{enumerate}
\end{proof}

From \autoref{L:gwpo-spo} and \autoref{L:spo-gwpo}, we establish the simulation result.

\begin{thm}
\label{T:gwpo-spo}
Let $>_\gwpo$ be the generalized weighted path order
induced by stable order pairs $({\sqsupseteq_A}, {\sqsupset_A})$ and $({\sqsupseteq_B}, {\sqsupset_B})$
with the subterm property of $({\sqsupseteq_A}, {\sqsupset_A})$.
Then $>_\gwpo$ and $>_\spo$ are identical,
where $>_\spo$ is induced by $({\sqsupseteq_{AB}}, {\sqsupset_{AB}})$. \qed
\end{thm}

As a consequence, $>_\gwpo$ and $>_\spo$ coincide under the preconditions of \autoref{T:gwpo}.
This proves \autoref{T:gwpo} via \autoref{T:spo}.
Moreover, the simulation result can be easily extended to the monotonic versions
by using lexicographic combination of reduction triples.

\begin{prop}
\label{P:lex-reduction-triple}
If $({\geqslant_A}, {\sqsupseteq_{A}}, {\sqsupset_{A}})$ and 
$({\geqslant_B}, {\sqsupseteq_{B}}, {\sqsupset_{B}})$ are reduction triples,
so is the triple $({\geqslant_A \cap \geqslant_B}, {\sqsupseteq_{AB}}, {\sqsupset_{AB}})$.
\end{prop}
\begin{proof}
For the harmony property, let $s \geqslant_\AA t$ and $s \geqslant_\BB t$,
and let $C = f(\bar{u}, \hole, \bar{v})$.
Since $\sqsupseteq_A$ and $\sqsupseteq_B$ are harmonious with respect to $\geqslant_A$ and $\geqslant_B$, respectively,
we have $C[s] \sqsupseteq_A C[t]$ and  $C[s] \sqsupseteq_B C[t]$.
So, we obtain $C[s] \sqsupseteq_{AB} C[t]$.
The other properties are readily confirmed.
\end{proof}

Essentially, \cite[Proposition~4.3]{BFR00} is the special case when
$\geqslant_A$ and $\geqslant_B$ are identical.

\begin{cor}
\label{C:mgwpo-mspo}
Let $>_\mgwpo$ be the monotone generalized weighted path order
induced by reduction triples $({\geqslant_A}, {\sqsupseteq_A}, {\sqsupset_A})$ and $({\geqslant_B}, {\sqsupseteq_B}, {\sqsupset_B})$
with the subterm property of $\sqsupseteq_A$.
Then $>_\mgwpo$ and $>_\mspo$ are identical,
where $>_\mspo$ is induced by $({\geqslant_A \cap \geqslant_B}, {\sqsupseteq_{AB}}, {\sqsupset_{AB}})$. \qed
\end{cor}

Together with \autoref{T:mspo} and \autoref{P:lex-reduction-triple},
this simulation result proves \autoref{T:mgwpo}.
We remark that some of the assumptions of \autoref{C:mgwpo-mspo} could be dropped.
For example, as in \autoref{T:gwpo-spo}, well-foundedness of the reduction triples is not necessary for simulation. 

It is worth examining this simulation result a little more.
The MGWPO (\autoref{T:mgwpo}) provides a way to turn two reduction triples, say $({\geqslant_A}, {\sqsupseteq_A}, {\sqsupset_A})$
and $({\geqslant_B}, {\sqsupseteq_B}, {\sqsupset_B})$, into a reduction order,
provided that $({\sqsupseteq_A}, {\sqsupset_A})$ has the subterm property.
If this property is not satisfied (cf.~\autoref{R:algebra-redtriple}), as in the proof of the simulation result,
we can instead use \autoref{P:lex-reduction-triple} to obtain the single reduction triple 
$({\geqslant_A \cap \geqslant_B}, {\sqsupseteq_{AB}}, {\sqsupset_{AB}})$
and then turn it into a reduction order via the MSPO (\autoref{T:mspo}).
Bypassing the simplicity requirement in this way gives us more termination proving power.

\begin{rem}
\label{R:frocos}
We add a few words about the preliminary paper~\cite{SH23} of this chapter.
There, the main result is a special case of \autoref{T:gwpo-spo} corresponding to the original 
WPO, namely when $A$ is induced by a weakly simple and monotone algebra and $B$ by a precedence (via \autoref{P:prec-redtriple}).
Based on this result and the observation above, a generalized version of the WPO is obtained
by using \autoref{P:algebra-redtriple} (with possibly non-weakly-simple algebras) for $A$.
Somewhat confusingly, although the generalization is merely a special version of the MSPO, 
they call it the GWPO. 
\end{rem}

The simulation result does not hold, for instance,
when the SPO is defined using multiset extension as in \cite{BFR00}
and the (G)WPO is defined using lexicographic extension
as in this article and in the original article of the WPO~\cite{YKS15}.
This is why our definitions of the path orders are aligned with lexicographic extension.

\subsection{GWPO as Optimized SPO}
\label{SS:gwpo-as-optimized-spo}
Another way of seeing the simulation result is an optimization technique of term comparison. 
Unlike the SPO, the GWPO allows us to avoid
recursive comparison with arguments of $t$ if $s \sqsupset_A t$,
see \autoref{E:SK} and \autoref{E:spo}.
This is reminiscent of the relationship between the KBO and the WPO:
the former is a special case of the latter with the strict simplicity of algebras, where recursive comparison 
with arguments of $t$ can be further avoided~\cite{YKS15}.
As one can expect, we also have a more KBO-like variant of GWPO
that assumes the subterm property of $\sqsupset_A$ and thereby avoids unnecessary comparison.

To present the KBO-like variant, we need a few definitions.
First, following \cite[Definition~3.2]{Yam14}, we extend admissibility of weight for the KBO to order pairs:
an order pair $({\sqsupseteq_A}, {\sqsupset_A})$ is \emph{admissible}
for another order pair $({\sqsupseteq_B}, {\sqsupset_B})$ 
if one of the following conditions holds for all $n$-ary function symbols $f$.
\begin{itemize}
\item $f(t_1, \ldots, t_n) \sqsupset_A t_i$ for all terms $t_1, \ldots, t_n$ and integers $1 \leqslant i \leqslant n$
\item $n = 1$ and $f(t) \sqsupseteq_B g(\bar{u})$ for all terms $t, \bar{u}$ and function symbols $g$
\end{itemize}
If the latter holds for a function symbol $f$, we call it a \emph{special unary symbol}.
By definition, if $\sqsupset_A$ has the subterm property then $A$ is admissible regardless of $B$.
Moreover, we say that an order pair $({\sqsupseteq}, {\sqsupset})$ is \emph{sensible}
if $\Var(s) \supseteq \Var(t)$ whenever $s \sqsupseteq t$.
This is a technical requirement to be explained later.

\begin{thm}
\label{T:kbo-like-gwpo}
Let $({\sqsupseteq_A}, {\sqsupset_A})$ and $({\sqsupseteq_B}, {\sqsupset_B})$ be order pairs 
and $>_\gwpo$ the GWPO induced by them.
Suppose that $A$ has the subterm property, and $A$ is sensible and admissible for $B$.
Then $s >_\gwpo t$ if and only if
\begin{enumerate}[1.]
  \item \label{kbogwpo:1} $s \sqsupset_A t$, or
  \item \label{kbogwpo:2} {$s = f(\seq[m]{s}) \sqsupseteq_A t$ and one of the following conditions holds:
  \begin{enumerate}[A.]
    \item \label{kbogwpo:2A} $t$ is a variable\footnote{%
    {This case sits oddly but goes back to \cite[Definition~3.1]{Yam14}.}}
    \item \label{kbogwpo:2B} $t = g(\seq[n]{t})$ and
    \begin{enumerate}[(I)]
      \item \label{kbogwpo:2B(I)} $s \sqsupset_B t$, or
      \item \label{kbogwpo:2B(II)} $s \sqsupseteq_B t$ and $(\seq[m]{s}) >_\gwpo^\lex (\seq[n]{t})$
    \end{enumerate}
  \end{enumerate}}
\end{enumerate}
\end{thm}
\begin{proof}
For the ``if'' direction, it suffices to show that each of case~\KBOGWPO{1}, case~\KBOGWPO{2A} and case~\KBOGWPO{2B} entails $s >_\gwpo t$.
If case~\KBOGWPO{1} holds then trivially $s >_\gwpo t$. So let $s = f(\seq[m]{s}) \sqsupseteq_A t$. 
\begin{enumerate}[A.]
\item Assume that $t$ is a variable. Since $A$ is sensible, $t$ occurs in $s$.
So $s >_\gwpo t$ follows from the subterm property.
\item Assume that $t = g(\seq[n]{t})$.
\begin{enumerate}[I.]
\item Suppose $s \sqsupset_B t$. Then $g$ is not a special unary symbol,
because otherwise we obtain a contradiction from $t \sqsupseteq_B s \sqsupset_B t$.
So $s \sqsupseteq_A t \sqsupset_A t_j$ and therefore $s >_\gwpo t_j$ for all $j$. Hence, $s >_\gwpo t$ by case~\GWPO{2b(i)}.  
\item Suppose $s \sqsupseteq_B t$ and $(\seq[m]{s}) >_\gwpo^\lex (\seq[n]{t})$. If $g$ is not a special unary symbol, then
by the same argument as the previous case we can show that $s >_\gwpo t$ by case~\GWPO{2b(ii)}.
So assume that $g$ is a special unary symbol.
From $(\seq[m]{s}) >_\gwpo^\lex (t_1)$ we obtain $s_1 \geqslant_\gwpo t_1$.
Since $s >_\gwpo s_1$ by the subterm property,
we now have $s >_\gwpo t_1$ by transitivity.
Hence, $s >_\gwpo t$ by case~\GWPO{2b(ii)}.
\end{enumerate}
\end{enumerate}
For the ``only-if'' direction, it suffices to show that each case of $s >_\gwpo t$ entails
one of case~\KBOGWPO{1}, case~\KBOGWPO{2A}, and case~\KBOGWPO{2B}.
Trivially, case~\GWPO{1} entails case~\KBOGWPO{1}, case~\GWPO{2b(i)} entails case~\KBOGWPO{2B(I)}, and case~\GWPO{2b(ii)} entails case~\KBOGWPO{2B(II)}.
So assume that case~\GWPO{2a} holds, namely $s = f(\seq[m]{s}) \sqsupseteq_A t$ and $s_i \geqslant_\gwpo t$ for some $i$.
If $t$ is a variable, then case~\KBOGWPO{2A} applies. So let $t = g(\seq[n]{t})$.
If $f$ is not a special unary symbol, then we have $s \sqsupset_A s_i \sqsupseteq_A t$ or $s \sqsupset_A s_i \sqsupset_A t$,
and therefore case~\KBOGWPO{1} applies.
Otherwise, $s \sqsupseteq_B t$ and $m = i = 1$.
Since $t >_\gwpo t_1$ follows from the subterm property,
we obtain $s_1 >_\gwpo t_1$ and thereby
$(s_1) >_\gwpo^\lex (\seq[n]{t})$. Hence, case~\KBOGWPO{2B(II)} holds.
\end{proof}

\begin{exa}[continued from \autoref{E:SK}]
Since
$x \cdot_\AA y = x + y + 1$ is strictly simple,
$>_\AA$ satisfies the subterm property, and thus
$({\geqslant_\AA},{>_\AA})$ is admissible for 
$({\geqslant_\AA^\sharp},{>_\AA^\sharp})$.  So \autoref{T:kbo-like-gwpo} is
available.  Recall the derivation of 
\(
t = (x \cdot z) \cdot^\sharp (y \cdot z) >_\gwpo
((\m{S} \cdot x) \cdot y) \cdot^\sharp z = u
\).
As $t \geqslant_\AA u$ and $t >_\AA^\sharp u$, it is now derived by 
case~\KBOGWPO{2B(I)} of the theorem, without checking 
$t >_\gwpo (\m{S} \cdot x) \cdot y$ or $t >_\gwpo z$.
\end{exa}

\autoref{T:kbo-like-gwpo} allows us to improve the complexity of deciding $s >_\gwpo t$ in the following sense:
Consider deciding $s >_\gwpo t$ for fixed order pairs $A$ and $B$.
If we measure the number of comparisons by $A$ and $B$,
the best algorithm known so far (simply following that for the LPO~\cite{Loc06LPO}) requires
a number of comparisons quadratic in $|s| + |t|$.
If fortunately $A$ is sensible and admissible for $B$, \autoref{T:kbo-like-gwpo} suggests that
the number of comparisons can be reduced to linear in $|s| + |t|$.
This is an analog of the fact that KBO comparison $s >_\kbo t$  admits an implementation
of linear time complexity (with respect to $|s| + |t|$)~\cite{Loc06KBO}.
We note that the simulation result (\autoref{T:gwpo-spo}) allows us to apply the same optimization 
to the SPO induced by the lexicographic combination $({\sqsupseteq_{AB}}, {\sqsupset_{AB}})$.

The requirement for $A$ to be sensible cannot be dropped from
\autoref{T:kbo-like-gwpo}.

\begin{exa}
\label{E:sensible}
Consider the signature $\FF = \{ \m{c}^{(0)} \}$. 
Let $A$ and $B$ be the trivial order pair $({\TT \times \TT}, {\emptyset})$.
Since $\prsuperterm$ is the empty relation $\emptyset$ in this signature,
$A$ has the subterm property and is admissible for $B$, but it is not sensible.
However, the GWPO $>_\gwpo$ induced by $A$ and $B$ does not satisfy $\m{c} >_\gwpo x$
while case~\KBOGWPO{2A} follows from $\m{c} \sqsupseteq_A x$.
This pathological example also serves as a
counterexample to \cite[Theorem~3.3]{Yam14} by considering the algebra
whose carrier
is the set of all terms and ordered by the reduction pair $A$.
\end{exa}

A precursor of the observation underlying Theorems~\ref{T:gwpo-spo} and~\ref{T:kbo-like-gwpo}
appears in the technical report by
Geser~\cite{Ges92}. More precisely, \cite[Theorem~5]{Ges92} states that a variant of KBOs (which he
attributes to Dershowitz~\cite{Der82}) can be simulated by SPOs.
For comparison, we present his observation in our order-pair based formulation.
We first define our version of Dershowitz's KBO.

\begin{defi}
\label{D:Dershowitz}
Let $({\sqsupseteq}, {\sqsupset})$ be an order pair on terms.
The relation $>_\Dershowitz$ is inductively defined as follows:
$s >_\Dershowitz t$ if
\begin{enumerate}[1.]
\item
\label{dershowitz:1}
$s \sqsupset t$, or
\item
\label{dershowitz:2}
$s = f(\seq[m]{s}) \sqsupseteq t$ and one of the following conditions holds.
\begin{enumerate}[a.]
\item
\label{derwotiz:2a}
$t$ is a variable
\item
\label{dershowitz:2b}
$t = g(\seq[n]{t})$ and $(\seq[m]{s}) >_\Dershowitz^\lex (\seq[n]{t})$
\end{enumerate}
\end{enumerate}
\end{defi}

One can find several differences between \autoref{D:Dershowitz} and 
Geser's original formulation (\cite[Definition~8.1]{Ges92}).
There, the order $>_\Dershowitz$ is formulated as a quasi-order parameterized by another quasi-order;
case~\DER{2} demands $s \equiv t$ for the equivalence part $\equiv$ of $\sqsupseteq$; and
case~\DER{2b} allows a more general status extension including multiset extension.
(The same goes for comparing his definition of the SPO with ours.)
Now, Geser's observation is stated as follows.

\begin{cor}
\label{C:kbo-like-spo}
Let $({\sqsupseteq}, {\sqsupset})$ be a stable and sensible order pair with
$\prsuperterm \subseteq {\sqsupset}$.
The relations $>_\spo$ and $>_\Dershowitz$ are identical
when both are induced by $({\sqsupseteq}, {\sqsupset})$.
\end{cor}
\begin{proof}
From  $\prsuperterm \subseteq {\sqsupset}$
the subterm property of $({\sqsupseteq}, {\sqsupset})$ is guaranteed.
So, by \autoref{T:gwpo-spo}
$>_\spo$ is identical to the GWPO $>_\gwpo$ induced by $({\sqsupseteq}, {\sqsupset})$ and
the trivial order pair $(\TT \times \TT,\emptyset)$.
Thanks to ${\prsuperterm} \subseteq {\sqsupset}$,
the order pair $({\sqsupseteq}, {\sqsupset})$ is admissible for $(\TT \times \TT,\emptyset)$.
So we conclude by observing the correspondence between case~\DER{2b} of \autoref{D:Dershowitz} and
case~\KBOGWPO{2B} of \autoref{T:kbo-like-gwpo} with $B$ trivial.
\end{proof}

We note that the subterm property of the strict order $\sqsupset$
cannot be replaced with that of the quasi-order $\sqsupseteq$ in \autoref{C:kbo-like-spo}.
Even worse, $>_\Dershowitz$ may not be well-founded.

\begin{exa}
\label{E:geser}
Let $({\sqsupseteq}, {\sqsupset})$ be the trivial order
pair $({\TT \times \TT}, {\emptyset})$.
The corresponding relation $>_\Dershowitz$ admits the
infinite sequence:
\[
\m{f}(\m{g}(x), x)
>_\Dershowitz
\m{f}(x, \m{f}(\m{g}(x), x))
>_\Dershowitz
\m{f}(x, \m{f}(x, \m{f}(\m{g}(x), x)))
>_\Dershowitz
\cdots
\]
Here \DER{2b} is repeatedly applied together with $\m{g}(x) >_\Dershowitz x$.
Also, $>_\Dershowitz$ is not identical to any SPO $>_\spo$, because
$\m{f}(x,\m{f}(\m{g}(x),x)) >_\spo \m{f}(\m{g}(x),x)$ holds by definition.
\end{exa}

In fact, the same example shows that \cite[Theorems~4 and~5]{Ges92}, which
state that his version of $>_\Dershowitz$ is well-founded and coincides
with his SPO, do not hold as stated.  The subterm property should instead
be assumed for $\sqsupset$ rather than $\sqsupseteq$.

\subsection{Ground Totality}
\label{SS:ground-totality}

A strict order $>$ on terms is \emph{ground total} if 
either $s > t$ or $t > s$ holds for all \emph{ground} terms $s \neq t$.
Here, ground terms are terms that contain no variable.
Ground totality of reduction orders is of interest since refutational completeness of
theorem proving procedures like ordered completion~\cite{BDP89} relies on it.
This subsection investigates the property for the (G)WPO and the SPO.

We begin with preliminary definitions.
We say that a strict order $>$ is \emph{total} on a set $X$ if
either $x > y$ or $y > x$ holds for all different elements $x, y \in X$.
(So, ground totality is totality restricted to ground terms.) 
We say that a quasi-order $\geqslant$ is \emph{total} on $X$ if its strict part is total.
An order pair $({\geqslant}, {>})$ is \emph{almost total} on $X$ if
$x > y$, $x \equiv y$ or $y > x$ holds for all $x, y \in X$.
Here, $\equiv$ is the equivalence part ${\geqslant} \cap {\leqslant}$ of $\geqslant$.
For example, 
the usual orders $\geqslant_\NN$ and $>_\NN$ on $\NN$ are total,
and also the order pair $({\geqslant_\NN}, {>_\NN})$ is almost total.

We say an algebra $\AA$ is \emph{almost total} if
the underlying order pair $({\geqslant}, {>})$ is almost total on the domain.
Polynomial interpretations and max/plus interpretations are such examples.
The order pair $({\geqslant_\AA}, {>_\AA})$ induced by an almost total algebra $\AA$ is almost total on ground terms.
We note that $>_\AA$ may not be ground total even when $>$ is total:
for the polynomial interpretation $\AA$ of \autoref{E:algebra} with $\m{S}_\AA = \m{K}_\AA = 1$
we have neither $\m{S} >_\AA \m{K}$ nor $\m{K} >_\AA \m{S}$.

The WPO gives rise to a ground-total reduction order, under reasonable assumptions.

\begin{thm}
\label{T:wpo-ground-total}
Let $\AA$ be an algebra that is well-founded, normal, weakly simple and monotone,
and almost total.
Let $\succeq$ be a well-founded total precedence.
Then
the WPO $>_\wpo$ induced by $\AA$ and $\succeq$ is a ground-total reduction order.
\end{thm}

Although the fact above is easy to prove and implicitly used in
\cite{JK20}, the authors are not aware of a published proof.  Later we
prove it by a more general result.

Unfortunately, in contrast to WPOs, MGWPOs may not be ground total, even
when they are constructed from total algebras and precedences.  The same
goes for MSPOs by the simulation result (\autoref{C:mgwpo-mspo}).

\begin{exa}[adapted from \textnormal{\cite{Zan94}}]
\label{E:mgwpo-non-ground-total}
Let $\FF = \{ \m{a}^{(0)}, \m{b}^{(0)}, \m{f}^{(1)}, \m{g}^{(1)} \}$, and
let $\AA$ be the linear polynomial interpretation for $\FF$ with
\begin{align*}
\m{a}_\AA &= 1
&
\m{b}_\AA &= 1
&
\m{f}_\AA(x) &= x + 1 
&
\m{g}_\AA(x) &= x + 1
\end{align*}
and $\BB$ the following linear polynomial interpretation for $\FF \cup \FF^\sharp$.
\begin{align*}
\m{a}_\BB &= 1
&
\m{b}_\BB &= 0
&
\m{f}_\BB(x) &= x
&
\m{g}_\BB(x) &= 0
&
\m{a}^\sharp_\BB &= 0
&
\m{b}^\sharp_\BB &= 1
&
\m{f}^\sharp_\BB(x) &= x
&
\m{g}^\sharp_\BB(x) &= 0
\end{align*}
The MGWPO $>_\mgwpo$ induced by $({\geqslant_\AA}, {\geqslant_\AA}, {>_\AA})$ and $({\geqslant_\BB}, {\geqslant_\BB^\sharp}, {>_\BB^\sharp})$
satisfies $\m{f}(\m{a}) >_\mgwpo \m{f}(\m{b})$ and $\m{g}(\m{b}) >_\mgwpo \m{g}(\m{a})$.
If $>_\mgwpo$ were ground total, then we would have either $\m{a} >_\mgwpo \m{b}$ or $\m{b} >_\mgwpo \m{a}$,
from which we obtain a contradiction using the monotonicity.
Even worse, by the same argument we can show that there is no ground-total reduction order 
that extends $>_\mgwpo$.
Therefore, combining $({\geqslant_\BB}, {\geqslant_\BB^\sharp}, {>_\BB^\sharp})$
with a precedence (via \autoref{P:prec-redtriple}) lexicographically does not resolve the issue.
\end{exa}

An observation here is that even though $\m{a}$ and $\m{b}$
are incomparable with respect to $>_\mgwpo$, we actually have $\m{b} >_\gwpo \m{a}$.
Indeed, we can show that $>_\gwpo$ is ground total under a few reasonable assumptions.
We say that $({\sqsupseteq}, {\sqsupset})$ \emph{discriminates root symbols}
if $f = g$ whenever $f(\bar{t}) \equiv g(\bar{t})$.
For instance, let $({\sqsupseteq_P}, {\sqsupset_P})$ be an order pair
induced by a precedence $\succeq$ (as in \autoref{P:prec-redtriple}).  If
$\succeq$ is total, then $({\sqsupseteq_P}, {\sqsupset_P})$ discriminates
root symbols. With this in mind, \autoref{T:wpo-ground-total} is proved
from \autoref{T:wpo-reduction-order} and the following lemma.

\begin{lem}
\label{L:gwpo-ground-totality}
Let $({\sqsupseteq_A},{\sqsupset_A})$ and $({\sqsupseteq_B},{\sqsupset_B})$
be order pairs that are almost total on ground terms.  If $B$ discriminates
root symbols, then the induced GWPO $>_\gwpo$ satisfies
$s >_\gwpo t$ or $t >_\gwpo s$ for all distinct ground terms $s$ and $t$.
\end{lem}
\begin{proof}
Let $s = f(\seq[m]{s})$ and $t = g(\seq[n]{t})$ be different ground terms.
We prove $s >_\gwpo t$ or $t >_\gwpo s$ by induction on $|s| + |t|$.
If $s \sqsupset_A t$ or $t \sqsupset_A s$ then we are done by case~\GWPO{1}.
So hereafter we assume $s \sqsupseteq_A t$ and $t \sqsupseteq_A s$
as $A$ is almost total on ground terms.
If $s_i \geqslant_\gwpo t$ or $t_i \geqslant_\gwpo s$ for some $i$,
then we are done by case~\GWPO{2a}.
Otherwise, by the induction hypothesis,
$s >_\gwpo t_j$ for all $1 \leqslant j \leqslant n$ and 
$t >_\gwpo s_i$ for all $1 \leqslant i \leqslant m$.
Then, we again use the assumption that $B$ is almost total on ground terms.
If $s \sqsupset_B t$ or $t \sqsupset_B  s$ then we are done by case~\GWPO{2b(i)}.
Otherwise, $s \sqsupseteq_B t$ and $t \sqsupseteq_B s$.
Since $B$ discriminates root symbols and $s \neq t$,
we have $(\seq[m]{s}) \neq (\seq[n]{t})$. By the induction hypothesis
together with the fact that lexicographic extension preserves totality,
we have $(\seq[m]{s}) >_\gwpo^\lex (\seq[n]{t})$ or $(\seq[n]{t}) >_\gwpo^\lex (\seq[m]{s})$.\footnote{%
As the proof suggests, the lemma does not hold if we adopt multiset extension in the definition of the GWPO.}
So case~\GWPO{2b(ii)} applies.
\end{proof}

Note that $({\geqslant_\BB^\sharp}, {>_\BB^\sharp})$ of \autoref{E:mgwpo-non-ground-total}
does not discriminate root symbols, since both $\m{f}(\m{b}) \geqslant_\BB^\sharp \m{g}(\m{b})$
and $\m{g}(\m{b}) \geqslant_\BB^\sharp \m{f}(\m{b})$ hold.
In such a case, we can simply combine it lexicographically with an order pair induced by a total precedence.
This leads to the following construction:
Let $\succeq$ be a well-founded and total precedence, and 
$\AA$ an  $(\FF \cup \FF^\sharp)$-algebra that is almost total, normal, weakly monotone, well-founded,
and weakly simple for $\FF$.
Define $({\sqsupseteq}, {\sqsupset})$ as the lexicographic combination of 
$({\geqslant_\AA^\sharp}, {>_\AA^\sharp})$ and the order pair induced by $\succeq$.
Consider the GWPO $>_\gwpo$ induced by $({\geqslant_\AA}, {>_\AA})$ and $({\sqsupseteq}, {\sqsupset})$.
Then we have:

\begin{thm}
\label{T:gwpo-ground-total}
The GWPO $>_\gwpo$ defined as above is a ground-total reduction order.
\end{thm}
\begin{proof}
To show that $>_\gwpo$ is a reduction order,
the key is that $({\geqslant_\AA}, {\geqslant_\AA}, {>_\AA})$ and $({\geqslant_\AA}, {\sqsupseteq}, {\sqsupset})$
are reduction triples with the subterm property of $\geqslant_\AA$.
Using the normality of $\AA$, we can show that $>_\gwpo$ is identical to 
the MGWPO induced by those reduction triples, and therefore a reduction order.
Finally, \autoref{L:gwpo-ground-totality} shows that $>_\gwpo$ is
ground total.
\end{proof}

We remark that we could show the termination of \autoref{E:SK} with a ground-total reduction order
by additionally using a total precedence as in \autoref{T:gwpo-ground-total}.
Also, we note that \autoref{L:gwpo-ground-totality} and \autoref{T:gwpo-ground-total}
can be used for the SPO via the simulation result (\autoref{T:gwpo-spo}).

\section{Reduction Pairs}
\label{S:reduction-pair}

While we have considered the reduction-order versions of the MSPO and the (M)GWPO so far,
the original WPO has been extended to reduction pairs
so it can be used more effectively in the dependency pair framework.
After recalling the dependency pair framework in \autoref{SS:dp},
we introduce a reduction-pair version of the GWPO in \autoref{SS:gwpo-redpair-simple}.
As a special case of it, in \autoref{SS:spo-redpair} we also present the SPO as a reduction pair,
and discuss simulation of GWPOs (which does not hold as in \autoref{S:reduction-order}).
In \autoref{SS:rp-refinements} we discuss how to use the reduction pairs more effectively.
Finally, \autoref{SS:experiment-reduction-pair}
presents experimental data of GWPOs in the context of the dependency pair framework.

\subsection{Dependency Pairs}
\label{SS:dp}

We recall the \emph{dependency pair framework}~\cite{AG00,GTS05LPAR,HM05,HM07}.
Let $\RR$ be a TRS\@.
We write $\DD_\RR$
for the set $\{ f \mid f(\seq{\ell}) \to r \in \RR \}$ of \emph{defined symbols}.
Other symbols are called \emph{constructor symbols}.
Given a term $t$ of the form $f(\seq[n]{t})$
with $f \in \DD_\RR$, we write $t^\sharp$ for $f^\sharp(\seq[n]{t})$.
Here $f^\sharp$ is a fresh $n$-ary function symbol corresponding to $f$.
The set of such terms $t^\sharp$ is denoted by $\TT^\sharp$.
The TRS $\DP(\RR)$ is defined as follows:
\[
\DP(\RR) =
\{ \ell^\sharp \to t^\sharp \mid \text{$\ell \to r \in \RR$, 
$r \superterm t$, $\rt(t) \in \DD_\RR$, and $\ell \nprsuperterm t$} \}
\]
We note that the condition  $\ell \nprsuperterm t$ is Dershowitz's refinement~\cite{Der04}.
Rules in $\DP(\RR)$ are called \emph{dependency pairs}.
\emph{Dependency pair problems} are pairs $(\PP,\RR)$ of TRSs with
$\PP \subseteq \TT^\sharp \times \TT^\sharp$ and
$\RR \subseteq \TT \times \TT$.
A dependency pair problem $(\PP,\RR)$ is \emph{finite} if
there exists no infinite sequence of the form
$s_1 \to_\RR^* t_1 \rto_\PP s_2 \to_\RR^* t_2 \rto_\PP \cdots$,
where each $s_i$ is terminating with respect to $\RR$.

\begin{thm}
\label{T:dp}
A TRS $\RR$ is terminating if and only if $(\DP(\RR),\RR)$ is finite.
\end{thm}

Whenever we apply the dependency pair framework to a concrete TRS,
we stick to the following notational convention:
for a function symbol $f$ in the original signature $\FF$,
the capitalized symbol $F$ (instead of $f^\sharp$) denotes the corresponding copy introduced by the dependency pair framework.
This is reasonable, e.g., when we use the dependency pair framework with the GWPO
and \autoref{P:algebra-redtriple},
both of which perform marking with $\sharp$.

\begin{exa}
\label{E:factorial}
Consider the fragment $\RR$ of the TRS encoding of factorial:
\begin{align*}
\m{p}(\m{s}(x)) & \xrightarrow{1} x
&
\m{f}(\m{s}(x)) & \xrightarrow{2} \m{s}(x) \times \m{f}(\m{p}(\m{s}(x)))
\end{align*}
The defined symbols are $\m{f}$ and $\m{p}$, so 
$\DP(\RR) = \{ \m{F}(\m{s}(x)) \to \m{F}(\m{p}(\m{s}(x))), \m{F}(\m{s}(x)) \to \m{P}(\m{s}(x)) \}$.
\end{exa}

Reduction pairs $({\geqslant}, {>})$ are common decreasing measures for the dependency pair framework.
The next theorem is known as \emph{reduction pair processor}.

\begin{thm}
\label{T:rp}
Let $({\geqslant},{>})$ be a reduction pair.  A dependency pair problem
$(\PP,\RR)$ with $\PP \cup \RR \subseteq {\geqslant}$ is finite if and only
if $(\PP \setminus {>},\RR)$ is finite.
\end{thm}

Since $({\emptyset}, \RR)$ is trivially finite,
the theorem above implies a simple method for showing termination by a reduction pair.

\begin{cor}
\label{C:DP}
A TRS $\RR$ is terminating if
there is a reduction pair $({\geqslant},{>})$ with $\RR \subseteq {\geqslant}$
and $\DP(\RR) \subseteq {>}$.
\end{cor}

\begin{exa}[continued from \autoref{E:factorial}]
Let $\AA$ be the max/plus interpretation defined by
$\m{f}_\AA(x) = \m{F}_\AA(x) = \m{s}_\AA(x) = x \times_\AA y = x + 1$ and 
$\m{p}_\AA(x) = \m{P}_\AA(x) = \max \{ 0, x - 1\}$.
Then we have $\RR \subseteq {\geqslant_\AA}$ and $\DP(\RR) \subseteq {>_\AA}$.
For instance, $\m{F}(\m{s}(x)) >_\AA \m{F}(\m{p}(\m{s}(x)))$
boils down to the valid inequality $x + 2 > x + 1$ for all natural numbers $x$.
The others are also verified in a similar fashion.
So, from \autoref{C:DP} we conclude the termination of $\RR$.
\end{exa}

If a monotone reduction pair is employed, \autoref{T:rp} can be strengthened in a way that
some rules in $\RR$ can also be removed.
This technique is known as \emph{rule removal processor},
or just rule removal.

\begin{thm}
\label{T:rule-removal}
Let $({\geqslant},{>})$ be a monotone reduction pair. A dependency pair problem
$(\PP,\RR)$ with $\PP \cup \RR \subseteq {\geqslant}$ is finite if and only
if $({\PP},{\RR \setminus {>}})$ is finite.
\end{thm}

\emph{Argument filtering}~\cite{AG00,KNT06} is a popular transformation for reduction pairs.
As the name suggests, this transformation
filters out arguments from a given term. Formally,
an \emph{argument filter} $\pi$ is a mapping that 
associates
each $n$-ary function symbol $f$ to an integer $i$ or a list
$[i_1,\ldots,i_m]$ of integers with $1 \leqslant i_1 < \cdots < i_m \leqslant n$.
If $\pi(f)$ is a list for every function symbol $f$, then the argument filtering
is called \emph{non-collapsing}.
Non-collapsing argument filters are also called \emph{partial statuses}.\footnote{%
There are versions of partial status that allow permutation and/or duplication of arguments~\cite{YKS15,TSSY20,Yam22IJCAR}.
Such a version could be considered with additional care of boundedness of arity.}
The \emph{argument filtering} $\hat{\pi}$ is defined on terms as follows:
\[
\hat{\pi}(t) =
\begin{cases}
t & \text{if $t$ is a variable}
\\
\hat{\pi}(t_i)
& \text{if $t = f(t_1, \ldots, t_n)$ and $\pi(f) = i$}
\\
f(\hat{\pi}(t_{i_1}),\ldots,\hat{\pi}(t_{i_m}))
& \text{if $t = f(\seq{t})$ and $\pi(f) = [\seq[m]{i}]$}
\end{cases}
\]
Note that by applying
$\hat{\pi}$, arities of function symbols may change, or some function symbols
may be even eliminated. So more precisely,
terms after argument filtering are of a different signature.
We write $\FF^\pi$ for it to distinguish from the original signature $\FF$.
Given a binary relation $R$ on terms, we write 
$s \mathrel{R}^\pi t$ if $\hat{\pi}(s) \mathrel{R} \hat{\pi}(t)$.

\begin{prop}
\label{P:af}
Let $({\geqslant}, {>})$ be a reduction pair and $\pi$ an argument filter.
Then $({\geqslant^\pi}, {>^\pi})$ is a reduction pair.
\end{prop}

Although \autoref{P:af} can be used with the WPO as a reduction order (\autoref{D:wpo-reduction-order}),
a dedicated version of the WPO as a reduction pair is known.
Let $\pi$ be a partial status.
We abuse $i \in \pi(f)$ to mean that $i$ occurs in the list $\pi(f)$.
{Similarly, we may identify the list $\pi(f)$ and the corresponding set,
because here we only consider an increasing list of argument positions for $\pi(f)$.}
If $\pi(f) = [1, \ldots, n]$ for all $f^{(n)} \in \FF$,
we call $\pi$ \emph{total}.
Similarly, if $\pi(f) = \emptylist$ for all $f \in \FF$,
we call $\pi$ \emph{empty}.
Note that when the signature contains only constants,
the unique partial status is both empty and total.
If $\pi(f) = [ i_1, \ldots, i_n ]$ with 
$i_1 \leqslant \cdots \leqslant i_n$ then
we write $(t_{i_1}, \ldots, t_{i_n})$ as $\pi(f)(t_1, \ldots, t_m)$.

Below, we define a reduction-pair version of the WPO,
reusing the notation \(>_\wpo\) for brevity. The same convention applies to the upcoming GWPO/SPO.

\begin{defi}
\label{D:wpo-reduction-pair}
Let $\pi$ be a partial status, $\AA$ an algebra, and $\succeq$ a precedence.
The \emph{weighted path order} $({\geqslant_\wpo}, {>_\wpo})$ is the pair of relations
on terms defined simultaneously as follows: $s \geqslant_\wpo t$ if
\begin{enumerate}[1.]
\item \label{wpoRP:1}
$s >_\AA t$, or
\item \label{wpoRP:2}
$s \geqslant_\AA t$ and one of the following conditions holds:
\begin{enumerate}[a.]
\item \label{wpoRP:2a}
$s = f(\seq[m]{s})$ and $s_i \geqslant_\wpo t$ for some $i \in
\pi(f)$.
\item \label{wpoRP:2b}
$s = f(\seq[m]{s})$, $t = g(\seq[n]{t})$, $s >_\wpo t_j$ for all 
$j \in \pi(g)$, and
\begin{enumerate}[(i)]
\item \label{wpoRP:2b(i)}
$f \succ g$ or
\item \label{wpoRP:2b(ii)}
$f \succeq g$ and
$\pi(f)(\seq[m]{s}) \geqslant_\wpo^\lex \pi(g)(\seq[n]{t})$.
\end{enumerate}
\item \label{wpoRP:2c}
$s \in \VV$ and $s = t$
\end{enumerate}
\end{enumerate}
The relation $>_\wpo$ is defined by cases~\WPORP{1}, \WPORP{2a} and \WPORP{2b}
where $\geqslant_\wpo^\lex$ is replaced by $>_\wpo^\lex$.
Here $({\geqslant_\wpo^\lex}, {>_\wpo^\lex})$ is the lexicographic extension of
$({\geqslant_\wpo}, {>_\wpo})$.
\end{defi}

We say that an algebra $\AA$ is \emph{weakly (resp.\ strictly) $\pi$-simple}
if the argument position $i$ is weakly (resp.\ strictly) simple for all
function symbols $f$ and $i \in \pi(f)$.

\begin{thmC}[\cite{YKS15}]
\label{T:wpo-reduction-pair}
Let $\pi$ be a partial status, $\succeq$ a well-founded precedence,
and $\AA$ an algebra that is well-founded, normal, weakly $\pi$-simple and weakly monotone.
Then $({\geqslant_\wpo}, {>_\wpo})$ is a reduction pair.
If in addition $\pi$ is total, then  $({\geqslant_\wpo}, {>_\wpo})$ is a monotone reduction pair.
\end{thmC}

The assumption of weak $\pi$-simplicity is essential for well-foundedness; see~\autoref{E:wpo-weak-simple}.

\begin{exa}[Continued from \autoref{E:factorial}]
\label{E:factorial-wpo}
To show the termination of $\RR$, consider the linear polynomial interpretation $\AA$
with 
\begin{align*}
\m{s}_\AA(x) &= x + 1
&
x \times_\AA y &= x
&
\m{f}_\AA(x) = \m{p}_\AA(x) = \m{F}_\AA(x) = \m{P}_\AA(x) &= x
\end{align*}
and the partial status $\pi$ with
$\pi(\m{p}) = \emptylist$ and $\pi(g) = [1]$ for other function symbols $g$.
The WPO $({\geqslant_\wpo}, {>_\wpo})$ induced from $\AA$, $\pi$ and
the precedence with $\m{s} \succ \m{f} \succ \m{F} \succ \m{p}  \succ \m{P} \succ \times$
satisfies $\RR \subseteq {\geqslant_\wpo}$ and also $\DP(\RR) \subseteq {>_\wpo}$.
To elaborate,
$\m{p}(\m{s}(x)) \geqslant_\wpo x$ follows from case~\WPORP{1},
$\m{f}(\m{s}(x)) \geqslant_\wpo \m{s}(x) \times \m{f}(\m{p}(\m{s}(x)))$ from case~\WPORP{2b(i)},
and $\m{F}(\m{s}(x)) >_\wpo \m{P}(\m{s}(x))$ from case~\WPORP{2b(i)}.
The remaining inequality $\m{F}(\m{s}(x)) >_\wpo \m{F}(\m{p}(\m{s}(x)))$ is
verified as follows:
\begin{center}
\begin{prooftree}[separation=1em]
 \hypo{\m{F}(\m{s}(x)) \geqslant_\AA \m{F}(\m{p}(\m{s}(x)))}
  \hypo{\m{F}(\m{s}(x)) \geqslant_\AA \m{p}(\m{s}(x))}
  \hypo{\m{F} \succ \m{p}}
 \infer{2}[\WPORP{2b(i)}]{\m{F}(\m{s}(x)) >_\wpo \m{p}(\m{s}(x))}
  \hypo{\m{s}(x) \geqslant_\AA \m{p}(\m{s}(x))}
  \hypo{\m{s} \succ \m{p}}
 \infer{2}[\WPORP{2b(i)}]{\m{s}(x) >_\wpo \m{p}(\m{s}(x))}
\infer{3}[\WPORP{2b(ii)}]{\m{F}(\m{s}(x)) >_\wpo \m{F}(\m{p}(\m{s}(x)))}
\end{prooftree}
\end{center}
Here $\pi(\m{p}) = \emptylist$ is used for avoiding the undesired
comparison $\m{s}(x) >_\wpo \m{s}(x)$ for the second occurrence of
\WPORP{2b(i)}.  Using \autoref{C:DP} we conclude the termination of $\RR$.
\end{exa}

\subsection{GWPO as Reduction Pair}
\label{SS:gwpo-redpair-simple}

As in the reduction-order setting, 
a reduction-pair version of the GWPO is obtained by 
generalizing the parameters of the WPO to order pairs.

\begin{defi}
\label{D:gwpo-reduction-pair}
Let $\pi$ be a partial status, and
let $({\sqsupseteq_A}, {\sqsupset_A})$ and $({\sqsupseteq_B}, {\sqsupset_B})$ be pairs of relations on terms.
The \emph{generalized weighted path order}
$({\geqslant_\gwpo}, {>_\gwpo})$ is a pair of relations on terms defined simultaneously as follows:
$s \geqslant_\gwpo t$ if 
\begin{enumerate}[1.]
\item \label{gwpoRP:1} $s \sqsupset_A t$ or
\item \label{gwpoRP:2} $s \sqsupseteq_A t $ and one of the following conditions holds.
\begin{enumerate}[a.]
\item \label{gwpoRP:2a}
$s = f(\seq[m]{s})$ and $s_i \geqslant_\gwpo t$ for some $i \in
\pi(f)$.
\item \label{gwpoRP:2b}
$s = f(\seq[m]{s})$, $t = g(\seq[n]{t})$, $s >_\gwpo t_j$ for all 
$j \in \pi(g)$, and
\begin{enumerate}[(i)]
\item \label{gwpoRP:2b(i)}
$s \sqsupset_B t$ or
\item \label{gwpoRP:2b(ii)}
$s \sqsupseteq_B t$ and
$\pi(f)(\seq[m]{s}) \geqslant_\gwpo^\lex \pi(g)(\seq[n]{t})$.
\end{enumerate}
\item \label{gwpoRP:2c}
$s \in \VV$ and $s = t$
\end{enumerate}
\end{enumerate}
The relation $>_\gwpo$ is defined by cases~\GWPORP{1}, \GWPORP{2a} and \GWPORP{2b}
where $\geqslant_\gwpo^\lex$ is replaced by $>_\gwpo^\lex$.
Here $({\geqslant_\gwpo^\lex}, {>_\gwpo^\lex})$ is the lexicographic extension of
$({\geqslant_\gwpo}, {>_\gwpo})$.
\end{defi}

As is done for the original weighted path order in \cite[Section~4.2]{YKS15},
it is possible to incorporate more refinements in \autoref{D:gwpo-reduction-pair}.
This is considered in the first author's dissertation~\cite{Sai26}, see~\autoref{S:conclusion} for a brief discussion.

As in the reduction-order setting, we define a monotonic version of the generalized weighted path order.

\begin{defi}
\label{D:mgwpo-reduction-pair}
Let $\pi$ be a partial status, and
let $({\geqslant_A}, {\sqsupseteq_A}, {\sqsupset_A})$ and $({\geqslant_B}, {\sqsupseteq_B}, {\sqsupset_B})$ be triples of relations on terms.
The \emph{monotonic generalized weighted path order}
$({\geqslant_\mgwpo}, {>_\mgwpo})$
is defined as the pair of
${{\geqslant_A} \cap {\geqslant_B} \cap {\geqslant_\gwpo}}$ and ${{\geqslant_A} \cap {\geqslant_B} \cap {>_\gwpo}}$
where $({\geqslant_\gwpo}, {>_\gwpo})$ is induced by  $({\sqsupseteq_A}, {\sqsupset_A})$ and $({\sqsupseteq_B}, {\sqsupset_B})$.
\end{defi}

We say that
an order pair $({\sqsupseteq}, {\sqsupset})$ on terms is \emph{simple with respect to} 
a partial status $\pi$
if $f(t_1, \ldots, t_n) \sqsupseteq t_i$ or $f(t_1, \ldots, t_n) \sqsupset t_i$
holds for all function symbols $f^{(n)}$, terms $\seq[n]{t}$ and $i \in \pi(f)$.
This is an analog of simplicity of algebras with respect to partial statuses
which is demanded for the reduction-pair version of WPO\@.
Here we arrive at our main theorem.

\begin{thm}
\label{T:gwpo-reduction-pair}
Let $\pi$ be a partial status, and
$({\geqslant_A}, {\sqsupseteq_A}, {\sqsupset_A})$ and $({\geqslant_B}, {\sqsupseteq_B}, {\sqsupset_B})$ reduction triples.
If $({\sqsupseteq_A}, {\sqsupset_A})$ is simple with respect to $\pi$,
then $({\geqslant_\mgwpo}, {>_\gwpo})$ is a reduction pair.
\end{thm}

In fact, the proof is easily done by replaying that for the original WPO\@.
For completeness, the proofs are included in \autoref{S:omitted}.

The original weighted path order as a reduction pair (\autoref{T:wpo-reduction-pair})
is obtained by instantiating
$({\geqslant_A}, {\sqsupseteq_A}, {\sqsupset_A})$ with $({\geqslant_\AA}, {\geqslant_\AA}, {>_\AA})$,
and $({\geqslant_B}, {\sqsupseteq_B}, {\sqsupset_B})$ with the reduction pair induced by precedence.
As we expect,
the weak $\pi$-simplicity of $\AA$ guarantees that $({\sqsupseteq_A}, {\sqsupset_A})$ is simple with respect to $\pi$.

Now, let us look at a termination proof by \autoref{T:gwpo-reduction-pair}.

\begin{exa}
\label{E:gwpo-quot}
Consider the TRS $\RR$ taken from \cite[Example~3.4]{AG01}:
\begin{align*}
\m{m}(x, \m{0}) & \xrightarrow{1} \m{0}
&
\m{p}(\m{0}, y) & \xrightarrow{4} y
&
\m{q}(\m{0}, \m{s}(y)) & \xrightarrow{6} \m{0}
\\
\m{m}(\m{s}(x), \m{s}(y)) & \xrightarrow{2} \m{m}(x, y)
&
\m{p}(\m{s}(x), y) & \xrightarrow{5} \m{s}(\m{p}(x, y))
&
\m{q}(\m{s}(x), \m{s}(y)) & \xrightarrow{7} \m{s}(\m{q}(\m{m}(x, y), \m{s}(y)))
\\
\m{m}(\m{m}(x, y), z) & \xrightarrow{3} \m{m}(x, \m{p}(y, z))
\end{align*}
Here $\m{m}$ stands for minus, $\m{p}$ plus, and $\m{q}$ quotient.
The set $\DP(\RR)$ consists of six rules:
\begin{align*}
\m{M}(\m{s}(x), \m{s}(y)) & \xrightarrow{8} \m{M}(x, y)
&
\m{M}(\m{m}(x, y), z) & \xrightarrow{9} \m{M}(x, \m{p}(y, z))
&
\m{M}(\m{m}(x, y), z) & \xrightarrow{10} \m{P}(y, z)
\\
\m{P}(\m{s}(x), y) & \xrightarrow{11} \m{P}(x, y)
&
\m{Q}(\m{s}(x), \m{s}(y)) & \xrightarrow{12} \m{Q}(\m{m}(x, y), \m{s}(y))
&
\m{Q}(\m{s}(x), \m{s}(y)) & \xrightarrow{13} \m{M}(x, y)
\end{align*}
Let $\AA$ be the linear polynomial interpretation over $\NN$ defined by
\begin{align*}
\m{0}_\AA &= 0
&
\m{s}_\AA(x) &= x + 1
&
\m{0}_\AA^\sharp &= 0
&
\m{s}^{\sharp}_\AA(x) &= 0
\\
\m{m}_\AA(x, y) &= x
&
\m{M}_\AA(x, y) &= x + 1
&
\m{m}^{\sharp}_\AA(x, y) &= x
&
\m{M}^{\sharp}_\AA(x, y) &= x
\\
\m{p}_\AA(x, y) &= x + y + 1
&
\m{P}_\AA(x, y) &= 0
&
\m{p}^{\sharp}_\AA(x, y) &= 1
&
\m{P}^\sharp_\AA(x, y) &= x
\\
\m{q}_\AA(x, y) &= x + 1
&
\m{Q}_\AA(x, y) &= x + 1
&
\m{q}^\sharp_\AA(x, y) &= 1
&
\m{Q}^{\sharp}_\AA(x, y) &= x
\end{align*}
and $\pi$ the partial status with $\pi(\m{M}) = \pi(\m{m}) = [1]$ and $\pi(f) = \emptylist$ for other function symbols $f$.
The order pair $({\geqslant_\AA}, {>_\AA})$ is simple with respect to $\pi$
since the first arguments of $\m{m}$ and $\m{M}$ are taken into account in $\AA$.
So, we construct the reduction pair $({\geqslant_\mgwpo}, {>_\gwpo})$ by \autoref{T:gwpo-reduction-pair}
with the reduction triples $({\geqslant_\AA}, {\geqslant_\AA}, {>_\AA})$ and 
$({\geqslant_\AA}, {\geqslant_\AA^\sharp}, {>_\AA^\sharp})$.
To conclude that $\RR$ is terminating, it suffices to show 
 $\RR \subseteq {\geqslant_\mgwpo}$ and $\DP(\RR) \subseteq {>_\gwpo}$.
This can be seen as follows. Since $\{ 1, 2, 4, 6, 8, 10, 12, 13 \} \subseteq {>_\AA}$, these are handled by case~\GWPORP{1}.
The other rules are oriented by $\geqslant_\AA$, so we compare them after marking the root symbols.
\begin{align*}
3\colon && \m{m}^\sharp(\m{m}(x, y), z) & \geqslant_\AA \m{m}^\sharp(x, \m{p}(y, z))
&
9\colon && \m{M}^\sharp(\m{m}(x, y), z) & \geqslant_\AA \m{M}^\sharp(x, \m{p}(y, z))
\\
5\colon && \m{p}^\sharp(\m{s}(x), y) & >_\AA \m{s}^\sharp(\m{p}(x, y))
&
11\colon && \m{P}^\sharp(\m{s}(x), y) & >_\AA \m{P}^\sharp(x, y)
\\
7\colon && \m{q}^\sharp(\m{s}(x), \m{s}(y)) & >_\AA \m{s}^\sharp(\m{q}(\m{m}(x, y), \m{s}(y)))
\end{align*}
For rules 3 and 9 we further apply recursive comparison by case~\GWPORP{2b(ii)}.
\end{exa}

In the example above, some arguments ignored by $\pi$ are 
used by the interpretation $\AA$, see for example $\m{p}_\AA(x, y) = x + y + 1$.
This is an advantage of having partial status built in
compared to applying the GWPO with the total status followed by argument filtering.
As with the case of the WPO~\cite[Proposition~1]{YKS15},
the latter with a non-collapsing argument filter can be simulated by the former,
although the two approaches are incomparable in general (as argument filtering allows collapsing).
Formally, let $\FF$ be the signature, $\pi$ a non-collapsing argument filter (therefore a partial status), and
$({\geqslant_A}, {\sqsupseteq_A}, {\sqsupset_A})$
and $({\geqslant_B}, {\sqsupseteq_B}, {\sqsupset_B})$ reduction triples on $\TT(\FF^\pi, \VV)$.
Define  $({\geqslant_{A'}}, {\sqsupseteq_{A'}}, {\sqsupset_{A'}})$ on $\TT(\FF, \VV)$ as 
$({\geqslant_A^\pi}, {\sqsupseteq_A^\pi}, {\sqsupset_A^\pi})$,
and similarly define $({\geqslant_{B'}}, {\sqsupseteq_{B'}}, {\sqsupset_{B'}})$ as 
$({\geqslant_B^\pi}, {\sqsupseteq_B^\pi}, {\sqsupset_B^\pi})$.

\begin{prop}
If $(\sqsupseteq_A, \sqsupset_A)$ has the subterm property,
then $(\sqsupseteq_{A'}, \sqsupset_{A'})$ is simple with respect to $\pi$.
Moreover,
$({\geqslant_{\mgwpo}^\pi}, {>_{\gwpo}^\pi})$ is identical to 
$({\geqslant_{\mgwpo'}}, {>_{\gwpo'}})$,
where $({\geqslant_{\mgwpo}}, {>_{\gwpo}})$ is induced by the total status and the reduction triples $A$ and $B$,
and $({\geqslant_{\mgwpo'}}, {>_{\gwpo'}})$ is induced by $\pi$ and the reduction triples $A'$ and $B'$.
\end{prop}
\begin{proof}
Our proof follows the same line as that of \cite[Proposition~1]{YKS15}:
The simplicity of $(\sqsupseteq_{A'}, \sqsupset_{A'})$ is easy to verify.
The other claim follows by showing that
$s \geqslant_{\gwpo}^\pi t$ if and only if $s \geqslant_{\gwpo'} t$
and $s >_{\gwpo}^\pi t$ if and only if $s >_{\gwpo'} t$  for all terms $s$ and $t$,
by induction on $|s| + |t|$.
\end{proof}

\subsection{SPO as Reduction Pair}
\label{SS:spo-redpair}

To the best of the authors' knowledge,
there is no existing work that formulates the semantic path order as a reduction pair
incorporating mutual recursion and partial status as in \autoref{D:gwpo-reduction-pair}.
Such a version of the semantic path order is naturally obtained from the GWPO\@.

\begin{defi}
\label{D:spo-reduction-pair}
Let $\pi$ be a partial status and 
let $({\sqsupseteq}, {\sqsupset})$ be a pair of relations on terms.
The \emph{semantic path order} $({\geqslant_\spo}, {>_\spo})$
induced by $\pi$ and $({\sqsupseteq}, {\sqsupset})$
is defined inductively as follows: $s \geqslant_\spo t$ if
\begin{enumerate}[a.]
\item \label{spoRP:a}
$s = f(\seq[m]{s})$ and $s_i \geqslant_\spo t$ for some $i \in
\pi(f)$.
\item \label{spoRP:b}
$s = f(\seq[m]{s})$, $t = g(\seq[n]{t})$, $s >_\spo t_j$ for all 
$j \in \pi(g)$, and
\begin{enumerate}[(i)]
\item \label{spoRP:b(i)}
$s \sqsupset t$ or
\item \label{spoRP:b(ii)}
$s \sqsupseteq t$ and
$\pi(f)(\seq[m]{s}) \geqslant_\spo^\lex \pi(g)(\seq[n]{t})$.
\end{enumerate}
\item \label{spoRP:c}
$s \in \VV$ and $s = t$
\end{enumerate}
The relation $>_\spo$ is defined by cases~\SPORP{a} and \SPORP{b},
where $\geqslant_\spo^\lex$ is replaced by $>_\spo^\lex$.
Here $({\geqslant_\spo^\lex}, {>_\spo^\lex})$ is the lexicographic extension of
$({\geqslant_\spo}, {>_\spo})$.
\end{defi}

Obviously, we could define $({\geqslant_\spo}, {>_\spo})$
as  $({\geqslant_\gwpo}, {>_\gwpo})$ induced from $({\TT \times \TT}, {\emptyset})$
and $({\sqsupseteq}, {\sqsupset})$.
The monotone version is defined as follows.

\begin{defi}
Let $\pi$ be a partial status, and
let $({\geqslant}, {\sqsupseteq}, {\sqsupset})$ be a triple of relations on terms.
The \emph{monotonic semantic path order}
$({\geqslant_\mspo}, {>_\mspo})$
is defined as the pair of
${{\geqslant} \cap {\geqslant_\spo}}$ and ${{\geqslant} \cap {>_\spo}}$
where $({\geqslant_\spo}, {>_\spo})$ is induced by  $({\sqsupseteq}, {\sqsupset})$.
\end{defi}

Again, we could define $({\geqslant_\mspo}, {>_\mspo})$
as  $({\geqslant_\mgwpo}, {>_\mgwpo})$ induced from the reduction triple $({\TT \times \TT}, {\TT \times \TT}, {\emptyset})$
and $({\geqslant}, {\sqsupseteq}, {\sqsupset})$.
This results in the following corollary.

\begin{cor}
\label{C:mspo-reduction-pair}
Let $\pi$ be a partial status, and
$({\geqslant}, {\sqsupseteq}, {\sqsupset})$ a reduction triple.
Then $({\geqslant_\mspo}, {>_\spo})$ is a reduction pair.
\qed
\end{cor}

The example below illustrates a key step in a termination proof 
of the benchmark \texttt{Saito\_26/AProVE\_08\_id\_inc\_p} in the Termination Problem Database~\cite{TPDB},
with some simplification.
As mentioned in \autoref{SS:experiment-reduction-pair},
the SPO plays a crucial role.

\begin{exa}
\label{E:inc}
Consider the set $\PP$ of the
dependency pairs
\begin{align*}
\m{G}(\m{c}(\m{s}(x),y))
&\to \m{G}(\m{c}(y,\m{p}(\m{s}(x))))
&
\m{G}(\m{c}(x,\m{s}(y)))
&\to \m{G}(\m{c}(\m{p}(\m{s}(y)),x))
\end{align*}
and the TRS $\RR = \{ \m{p}(\m{s}(x)) \to x \}$.
To show that the dependency pair problem $(\PP, \RR)$ is finite,
we use the max/plus interpretations $\AA$ and $\BB$ defined by:
\begin{align*}
\m{G}^{(\sharp)}_\AA(x) &= \max\{0, x - 1\}
&
\m{c}_\AA(x,y) &= \max\{x + 2, y + 1\}
&
\m{c}^\sharp_\AA(x,y) &= \max\{x, y + 1\}
\\
\m{p}^{(\sharp)}_\AA(x) &= \max\{0,x - 2\}
&
\m{s}^{(\sharp)}_\AA(x) &= x + 2
\\
\m{G}^{(\sharp)}_\BB(x) &= x + 1
&
\m{c}_\BB(x,y) &= y + 2
&
\m{c}^\sharp_\BB(x,y) &= x + 2
\\
\m{p}^{(\sharp)}_\BB(x) &= \max\{0,x - 2\}
&
\m{s}^{(\sharp)}_\BB(x) &= x + 2
\end{align*}
Let $({\geqslant}, {\sqsupseteq}, {\sqsupset})$ be the lexicographic combination of $({\geqslant_\AA}, {\geqslant_\AA^\sharp}, {>_\AA^\sharp})$ and 
$({\geqslant_\BB}, {\geqslant_\BB^\sharp}, {>_\BB^\sharp})$, and
let  $({\geqslant_\mspo}, {>_\spo})$ be the SPO induced by  $({\geqslant}, {\sqsupseteq}, {\sqsupset})$ and 
the partial status $\pi$ with
$\pi(\m{G}) = \pi(\m{p}) = \pi(\m{s}) = [1]$ and $\pi(\m{c})  = \emptylist$.
We can easily verify  $\RR \subseteq {\geqslant_\mspo}$. For the first dependency pair, we have the following derivation.
\begin{center}
\begin{prooftree}
    \hypo{\m{c}(\m{s}(x),y) \sqsupset \m{G}(\m{c}(y,\m{p}(\m{s}(x))))}
      \hypo{\m{c}(\m{s}(x),y) \sqsupset \m{c}(y,\m{p}(\m{s}(x)))}
    \infer{1}[\SPORP{b(i)}]{\m{c}(\m{s}(x),y) >_\spo \m{c}(y,\m{p}(\m{s}(x)))}
  \infer{2}[\SPORP{b(i)}]{\m{c}(\m{s}(x),y) \geqslant_\spo \m{G}(\m{c}(y,\m{p}(\m{s}(x))))}
\infer{1}[\SPORP{a}]{\m{G}(\m{c}(\m{s}(x),y)) >_\spo \m{G}(\m{c}(y,\m{p}(\m{s}(x))))}
\end{prooftree}
\end{center}
Here, $\m{c}(\m{s}(x),y) \sqsupset \m{G}(\m{c}(y,\m{p}(\m{s}(x))))$ follows from the decrease
\begin{align*}
\overbrace{\m{c}^\sharp(\m{s}(x),y)}^{\max \{ x + 2, y + 1 \}} & \geqslant_\AA \overbrace{\m{G}^\sharp(\m{c}(y,\m{p}(\m{s}(x))))}^{\max \{ y + 1, x\}}
&
\overbrace{\m{c}^\sharp(\m{s}(x),y)}^{x + 4} & >_\BB \overbrace{\m{G}^\sharp(\m{c}(y,\m{p}(\m{s}(x))))}^{x + 3}
\end{align*}
in $\BB$, and $\m{c}(\m{s}(x),y) \sqsupset \m{c}(y,\m{p}(\m{s}(x)))$ follows from the decrease
\[
\m{c}^\sharp_\AA(\m{s}_\AA(x),y) = \max \{ x + 2, y + 1 \} > \max \{ y, x + 1\} = \m{c}^\sharp_\AA(y,\m{p}_\AA(\m{s}_\AA(x)))
\]
in $\AA$.
Similarly, we have the following derivation for the second dependency pair.
\begin{center}
\begin{prooftree}
    \hypo{\m{G}(\m{c}(x,\m{s}(y))) \sqsupset \m{G}(\m{c}(\m{p}(\m{s}(y)),x))}
      \hypo{\m{G}(\m{c}(x,\m{s}(y))) \sqsupset \m{c}(\m{p}(\m{s}(y)),x)}
    \infer{1}[\SPORP{b(i)}]{\m{G}(\m{c}(x,\m{s}(y))) >_\spo \m{c}(\m{p}(\m{s}(y)),x)}
  \infer{2}[\SPORP{b(i)}]{\m{G}(\m{c}(x,\m{s}(y))) >_\spo \m{G}(\m{c}(\m{p}(\m{s}(y)),x))}
\end{prooftree}
\end{center}
Here, $\m{G}(\m{c}(x,\m{s}(y))) \sqsupset \m{G}(\m{c}(\m{p}(\m{s}(y)),x))$ follows from the decrease
\[
\m{G}^\sharp_\AA(\m{c}_\AA(x,\m{s}_\AA(y))) = \max \{ x + 1, y + 2 \} > \max \{ y + 1, x \} = \m{G}^\sharp_\AA(\m{c}_\AA(\m{p}_\AA(\m{s}_\AA(y)),x))
\]
in $\AA$, and $\m{G}(\m{c}(x,\m{s}(y))) \sqsupset \m{c}(\m{p}(\m{s}(y)),x)$ follows from the decrease
\begin{align*}
\overbrace{\m{G}^\sharp(\m{c}(x,\m{s}(y)))}^{\max \{ x + 1, y + 2 \} } & \geqslant_\AA \overbrace{\m{c}^\sharp(\m{p}(\m{s}(y)),x)}^{\max \{ y, x + 1 \}}
&
\overbrace{\m{G}^\sharp(\m{c}(x,\m{s}(y)))}^{y + 5} & >_\BB \overbrace{\m{c}^\sharp(\m{p}(\m{s}(y)),x)}^{y + 2}
\end{align*}
in $\BB$.
Hence, we conclude that $(\PP, \RR)$ is finite.
\end{exa}

Next, following \autoref{SS:simulation} we compare the generalized weighted path order with the semantic path order
as a reduction pair.
Let $({\geqslant_\gwpo}, {>_\gwpo})$ be the GWPO induced by
a partial status $\pi$ and stable order pairs $({\sqsupseteq_A}, {\sqsupset_A})$ and  $({\sqsupseteq_B}, {\sqsupset_B})$.
We consider the SPO $({\geqslant_\spo}, {>_\spo})$ induced by the lexicographic combination $({\sqsupseteq_{AB}}, {\sqsupset_{AB}})$ and $\pi$.
In this setting, an analog of \autoref{L:spo-gwpo} holds.

\begin{lem}
\label{L:spo-gwpo2}
Suppose that  $A$ is simple with respect to $\pi$.
For all terms $s$ and $t$,
it holds that $s \geqslant_\spo t$ implies $s \geqslant_\gwpo t$,
and $s >_\spo t$ implies $s >_\gwpo t$.
\end{lem}
\begin{proof}
We prove the two implications simultaneously by induction on $|s| + |t|$.
Here we only prove the first claim, as the other is analogously shown.
Let $s \geqslant_\spo t$. We analyze how it is derived.
\begin{enumerate}[a.]
\item If it follows from case~\SPORP{a}, then $s = f(\seq[m]{s})$ and $s_i \geqslant_\spo t$ for some $i \in \pi(f)$.
By the induction hypothesis we have $s_i \geqslant_\gwpo t$.
By definition, $s_i \sqsupset_A t$ or $s_i \sqsupseteq_A t$ holds.
In the former case, we use the simplicity assumption of $A$ to obtain
$s \sqsupseteq_A s_i \sqsupset_A t$ or $s \sqsupset_A s_i \sqsupset_A t$,
so $s \geqslant_\gwpo t$ by case~\GWPORP{1}.
In the latter case, we are done by case~\GWPORP{2a}.

\item If it follows from case~\SPORP{b}, we have $s = f(\seq[m]{s})$, $t = g(\seq[n]{t})$
and $s >_\spo t_j$ for all $j \in \pi(g)$.
By the induction hypothesis we have $s >_\gwpo t_j$ for all $j \in \pi(g)$.
\begin{enumerate}[(i)]
    \item If $s \sqsupset_{AB} t$, we further analyze how it follows:
    If $s \sqsupset_A t$ then we have $s >_\gwpo t$ by case~\GWPORP{1}.
    Otherwise, we conclude with case~\GWPORP{2b(i)}.
    \item If $s \sqsupseteq_{AB} t$, we further analyze how it follows:
    If $s \sqsupset_A t$ then we have $s >_\gwpo t$ by case~\GWPORP{1}.
     Otherwise, we conclude with case~\GWPORP{2b(ii)} using the induction hypothesis.
\end{enumerate}

\item If it follows from case~\SPORP{c}, we simply apply the corresponding case~\GWPORP{2c} with the reflexivity of $\sqsupseteq_A$. \qedhere
\end{enumerate}
\end{proof}

The assumption that $A$ is simple with respect to $\pi$ is essential for the lemma above.

\begin{exa}
\label{E:spo-gwpo2}
Consider the signature $\{ \m{a}^{(0)}, \m{f}^{(1)} \}$,
the partial status $\pi$ with $\pi(\m{f}) = [1]$,
and the algebra $\AA$ on $\NN$ defined by $\m{a}_\AA = 1$ and $\m{f}_\AA(x) = 0$.
Let $({\sqsupseteq_A}, {\sqsupset_A})$ be $({\geqslant_\AA}, {>_\AA})$ and 
let $({\sqsupseteq_B}, {\sqsupset_B})$ be $({\TT \times \TT}, {\emptyset})$.
Then we have $\m{f}(\m{a}) >_\spo \m{a}$ by case~\SPORP{a}, but $\m{f}(\m{a}) >_\gwpo \m{a}$ does not hold.
\end{exa}

The converse is shown by additionally assuming that $\pi$ is total.
The proof follows the same line as \autoref{L:gwpo-spo}.

\begin{lem}
\label{L:gwpo-spo2}
Suppose that $A$ is simple with respect to $\pi$, and $\pi$ is total.
For all terms $s$ and $t$,
it holds that $s \geqslant_\gwpo t$ implies $s \geqslant_\spo t$,
and $s >_\gwpo t$ implies $s >_\spo t$. \qed
\end{lem}

The assumption that $\pi$ is total cannot be dropped.
This contrasts with the simulation result (\autoref{T:gwpo-spo}) in the reduction-order setting.

\begin{exa}
\label{E:non-simulation}
Suppose that the signature only contains a unary function symbol $\m{f}$.
Let $\pi$ be the partial status defined by $\pi(\m{f}) = \emptylist$, 
let $\AA$ be the algebra on $\NN$ defined by $\m{f}_\AA(x) = x + 1$, and 
let $({\sqsupseteq_B}, {\sqsupset_B})$ be an arbitrary pair of relations on terms
(e.g., one induced by a precedence to consider the original WPO).
On one hand, the GWPO $({\geqslant_\gwpo}, {>_\gwpo})$ induced by $({\geqslant_\AA}, {>_\AA})$
and $({\sqsupseteq_B}, {\sqsupset_B})$ satisfies $\m{f}(x) \geqslant_\gwpo x$ and $\m{f}(x) >_\gwpo x$ 
by case~\GWPORP{1}. On the other hand, no SPO $({\geqslant_\spo}, {>_\spo})$ (regardless of the underlying order pair)
satisfies $\m{f}(x) \geqslant_\spo x$, because
the only possible case, case~\SPORP{a}, is blocked by $\pi(\m{f}) = \emptylist$.
\end{exa}

The cause behind the example above is very simple: if $s \sqsupset_A t$,
the GWPO satisfies $s \geqslant_\gwpo t$ and $s >_\gwpo t$ with case~\GWPORP{1} regardless of the underlying partial status,
while $s \geqslant_\spo t$ or $s >_\spo t$ may not be possible if the partial status does not allow it.

As a consequence of \autoref{L:spo-gwpo2} and \autoref{L:gwpo-spo2}, we can show an analog of \autoref{T:gwpo-spo} and \autoref{C:mgwpo-mspo}:
Let $({\geqslant_\mgwpo}, {>_\mgwpo})$ be the monotone generalized weighted path order
induced by a partial status $\pi$ and
reduction triples $({\geqslant_A}, {\sqsupseteq_A}, {\sqsupset_A})$ and $({\geqslant_B}, {\sqsupseteq_B}, {\sqsupset_B})$, and
let $({\geqslant_\mspo}, {>_\mspo})$ be the monotone semantic path order induced by $\pi$ and
$({\geqslant_A \cap \geqslant_B}, {\sqsupseteq_{AB}}, {\sqsupset_{AB}})$.

\begin{thm}
\label{T:mgwpo-mspo2}
If $({\sqsupseteq_A}, {\sqsupset_A})$ is simple with respect to $\pi$,
then we have ${\geqslant_\mspo} \subseteq {\geqslant_\mgwpo}$ and 
${>_\mspo} \subseteq {>_\mgwpo}$ as well as ${>_\spo} \subseteq {>_\gwpo}$. 
If additionally $\pi$ is total, then 
${\geqslant_\mspo}$, ${>_\mspo}$ and $>_\spo$
are identical to ${\geqslant_\mgwpo}$, ${>_\mgwpo}$ and $>_\gwpo$, respectively.
\qed
\end{thm}

So, in the setting of \autoref{T:mgwpo-mspo2} the GWPO is preferred over the SPO\@.
In other words, the SPO should be used only if the simplicity requirement is unsatisfied.

\autoref{T:mgwpo-mspo2} does not allow SPOs to simulate GWPOs with partial statuses.
So, it is natural to seek a better construction.
However, we note that any such construction must essentially assume well-foundedness of the order pair $A$,
in contrast to \autoref{T:gwpo-spo} and \autoref{T:mgwpo-mspo2}.
This can be seen from the following example.

\begin{exa}
\label{E:non-simulation2}
Suppose that the signature only contains a unary function symbol $\m{f}$.
Let $\pi$ be the partial status defined by $\pi(\m{f}) = \emptylist$, 
let $\AA$ be the algebra on non-positive integers
$\{ 0, -1, -2, \ldots \}$ defined by $\m{f}_\AA(x) = x - 1$, and 
let $({\sqsupseteq_B}, {\sqsupset_B})$ be an arbitrary pair of relations on terms.
Here, the carrier of $\AA$ is endowed with the ordinary order pair $({\geqslant}, {>})$ on integers,
so $>$ is not well-founded.
On one hand, the GWPO $({\geqslant_\gwpo}, {>_\gwpo})$ induced by $({\geqslant_\AA}, {>_\AA})$
and $({\sqsupseteq_B}, {\sqsupset_B})$ satisfies $x \geqslant_\gwpo \m{f}(x)$ and $x  >_\gwpo \m{f}(x)$ 
by case~\GWPORP{1} and $x >_\AA \m{f}(x)$.
On the other hand, no SPO (regardless of the underlying order pair and partial status)
satisfies $x \geqslant_\spo \m{f}(x)$, let alone $x >_\spo \m{f}(x)$,
as can be seen from the definition.
\end{exa}

We remark that non-well-founded algebras like the one above have been used for non-reachability proofs
with WPOs~\cite{Yam22IJCAR}.

\subsection{Using GWPO with Refinements to the Reduction Pair Processor}
\label{SS:rp-refinements}

There are many refinements to use reduction pairs more effectively in the dependency pair framework.
Among others, here we consider 
\emph{rule removal}, \emph{usable rule} and \emph{lexicographic combination},
providing some necessary facts to use GWPOs with them.

For rule removal~(\autoref{T:rule-removal}), we should use $({\geqslant_\mgwpo}, {>_\mgwpo})$ rather than $({\geqslant_\mgwpo}, {>_\gwpo})$,
because monotonicity is required for the reduction pair.

\begin{thm}
\label{T:gwpo-monotone-normal}
Let $\pi$ be a partial status, and
$({\geqslant_A}, {\sqsupseteq_A}, {\sqsupset_A})$ and $({\geqslant_B}, {\sqsupseteq_B}, {\sqsupset_B})$ reduction triples
with $({\sqsupseteq_A}, {\sqsupset_A})$ simple with respect to $\pi$.
Then $({\geqslant_\mgwpo}, {>_\mgwpo})$ is a normal reduction pair.
If in addition $\pi$ is total, then $({\geqslant_\mgwpo}, {>_\mgwpo})$  is a monotone reduction pair.
\end{thm}
\begin{proof}
The normality is a consequence of ${>_\gwpo} \subseteq {\geqslant_\gwpo}$, 
which is verified by straightforward induction.
The monotonicity is shown in a similar way to the proof of \autoref{T:mspo}.
The other properties are consequences of \autoref{T:gwpo-reduction-pair}.
\end{proof}

However, while exploring the experimental configurations for \autoref{SS:experiment-reduction-pair}
we found no case where rule removal with $({\geqslant_\mgwpo}, {>_\mgwpo})$ was useful.
Rather, we find that \autoref{T:gwpo-monotone-normal} is useful to construct a normal reduction pair,
which is required for lexicographic combination (to be explained soon).
The next example shows that the reduction pairs $({\geqslant_\mgwpo}, {>_\gwpo})$ and $({\geqslant_\mspo}, {>_\spo})$
are not always normal.

\begin{exa}
\label{E:gwpo-normal}
Let $\AA$ be the linear polynomial interpretation with
$\m{f}_\AA(x) = 0$ and $\m{a}_\AA = 1$.
The SPO $({\geqslant_\mspo}, {>_\spo})$ induced by the total status
$\pi$ with $\pi(\m{f}) = [1]$ and the reduction triple
$({\geqslant_\AA}, {\geqslant_\AA}, {>_\AA})$
is not normal.  In fact, it satisfies $\m{f}(\m{a}) >_\spo \m{a}$ but
$\m{f}(\m{a}) \ngeqslant_\mspo \m{a}$ due to
$\m{f}(\m{a}) \ngeqslant_\AA \m{a}$.
\end{exa}

\begin{rem}
\label{R:ro-rp}
As a side remark, we relate the reduction-pair version and the reduction-order version:
Let $>$ be the reduction-order version of the MGWPO (\autoref{D:mgwpo})
induced by reduction triples $A$ and $B$,
and let $({\geqslant_\mgwpo}, {>_\mgwpo})$ be the reduction-pair version of the MGWPO (\autoref{D:mgwpo-reduction-pair})
induced by the total status, $A$ and $B$.
Then by easy simultaneous induction one can show that ${\geqslant} \subseteq {\geqslant_\mgwpo}$
and ${>} \subseteq {>_\mgwpo}$, where $\geqslant$ is the reflexive closure of $>$.
For termination proving, the greater relations are the more profitable,
so we should use the reduction-pair version.
\end{rem}

The usable rule technique~\cite{Urb04,GTS05LPAR,GTS05FroCoS,GTS+06,HM07}
is a popular refinement to \autoref{T:rp} which relaxes the requirement $\RR \subseteq {\geqslant}$
to $\UU \subseteq {\geqslant}$ for a possibly smaller set $\UU \subseteq \RR$ of usable rules;
see also \cite[Section~4]{ST10} for the most general version.
The technique demands the reduction pair $({\geqslant}, {>})$ to be \emph{$\Ce$-compatible},
meaning that 
the reduction pair $({\geqslant}, {>})$ can be extended to another one $({\geqslant'}, {>'})$
with $\Ce \subseteq {\geqslant'}$.
Here $\Ce = \{ \m{c}(x, y) \to x, \m{c}(x, y) \to y \}$ and $\m{c}$ is a fresh binary symbol.
Like usual reduction pairs, the GWPO satisfies $\Ce$-compatibility almost for free,
as we show in what follows.

We say that a reduction triple $({\geqslant}, {\sqsupseteq}, {\sqsupset})$ is \emph{$\Ce$-compatible}
if it can be extended to another one $({\geqslant'}, {\sqsupseteq'}, {\sqsupset'})$ with $\Ce \subseteq {\geqslant'}$.
We call it \emph{strongly $\Ce$-compatible}
if such an extension can be chosen to satisfy both $\Ce \subseteq {\geqslant'}$ and $\Ce \subseteq {\sqsupseteq'}$.
Let $({\geqslant_A}, {\sqsupseteq_A}, {\sqsupset_A})$ and $({\geqslant_B}, {\sqsupseteq_B}, {\sqsupset_B})$
be reduction triples such that $({\sqsupseteq_A}, {\sqsupset_A})$ is simple with respect to a partial status $\pi$.

\begin{prop}
\label{P:gwpo-ce-compat}
If $A$ is strongly  $\Ce$-compatible and $B$ is $\Ce$-compatible,
then $({\geqslant_\mgwpo}, {>_\mgwpo})$ and $({\geqslant_\mgwpo}, {>_\gwpo})$ induced by $A$, $B$ and $\pi$ are $\Ce$-compatible reduction pairs.
\end{prop}
\begin{proof}
Using the assumption, we can 
extend $A$ and $B$ to $A'$ and $B'$, respectively, so that the inclusion
$\Ce \subseteq {{\geqslant_{A'}} \cap {\geqslant_{B'}} \cap {\sqsupseteq_{A'}}}$ holds.
We can also
extend $\pi$ to $\pi'$ with $\pi'(\m{c}) = [1, 2]$.
The relation $\geqslant'_\mgwpo$ induced by $A'$, $B'$ and $\pi'$ satisfies $\Ce \subseteq {\geqslant'_\mgwpo}$.
Here, $\Ce \subseteq {\geqslant'_\gwpo}$ follows from case~\GWPORP{2a} using $\Ce \subseteq {\sqsupseteq_{A'}}$.
Moreover, $\pi'$-simplicity for $A'$ is satisfied since $\Ce \subseteq {\sqsupseteq_{A'}}$ and $({\sqsupseteq_A}, {\sqsupset_A})$ is $\pi$-simple.
Therefore, \autoref{T:gwpo-reduction-pair} and \autoref{T:gwpo-monotone-normal} can be applied.
\end{proof}

The assumption of the proposition above is rather mild in practice, because usual reduction triples are $\Ce$-compatible:
the trivial reduction triple and those induced by precedences (\autoref{P:prec-redtriple})
are $\Ce$-compat\-ible;
usual interpretations (like polynomial interpretations, max/plus interpretations, matrix interpretation~\cite{EWZ08}, and so on)
can be used to construct $\Ce$-compatible reduction triples via \autoref{P:algebra-redtriple}, and
strongly $\Ce$-compatible ones via \autoref{P:redpair-redtriple}; and 
lexicographic combination (\autoref{P:lex-reduction-triple}) preserves $\Ce$-compatibility.
Moreover, the trivial reduction triple is strongly $\Ce$-compatible,
so \autoref{P:gwpo-ce-compat} can be applied to the SPO if it is induced by a $\Ce$-compatible reduction triple.
However, $\Ce$-compatibility is not guaranteed for the model-based construction of \autoref{P:model-redtriple}.
Note that $\Ce$ admits only trivial models~\cite{HM07}.

Continuing with the usable rule technique,
we consider \emph{$\pi$-compatibility} of reduction pairs,
which is relevant to the refinement called usable rule w.r.t.\ argument filtering~\cite{GTS+06}.
Let us recall the formal definition:
The $i$-th argument position of a function symbol $f$ is \emph{invariant} w.r.t.\ a relation $\leadsto$
if $C[s] \leadsto C[t]$ for all terms $s$ and $t$ and contexts $C = f(\bar{u}, \hole, \bar{v})$ with the hole $\hole$
at the $i$-th argument of $f$~\cite{SH25}. 
A reduction pair $({\geqslant}, {>})$ is $\pi$-compatible for a partial status $\pi$
if the $i$-th argument position of $f$ is invariant 
w.r.t.\ $\geqslant$
whenever $i \notin \pi(f)$.
So, below we give a sufficient condition of invariance 
w.r.t.\ $\geqslant_\mgwpo$.
For convenience, we extend the notion of invariance to reduction triples:
The $i$-th argument position of a function symbol $f$ is invariant w.r.t.\ a reduction triple $({\geqslant}, {\sqsupseteq}, {\sqsupset})$ if 
it is so with respect to $\geqslant$ and $\sqsupseteq$.

\begin{prop}
\label{P:gwpo-invariance}
Let $\geqslant_\mgwpo$ be the MGWPO induced by a partial status $\pi$ and
reduction triples $({\geqslant_A}, {\sqsupseteq_A}, {\sqsupset_A})$ and $({\geqslant_B}, {\sqsupseteq_B}, {\sqsupset_B})$.
Suppose that $({\sqsupseteq_A}, {\sqsupset_A})$ is simple w.r.t.\ $\pi$.
The $i$-th argument position of a function symbol $f$
is invariant w.r.t.\ $\geqslant_\mgwpo$ if it is so w.r.t.\ $A$ and $B$,
and also $i \notin \pi(f)$.
\end{prop}
\begin{proof}
Let $s$ and $t$ be terms and $C = f(\bar{u}, \hole, \bar{v})$ a context with the hole $\hole$
at the $i$-th argument of $f$.
By assumption, we have $C[s] \geqslant_A C[t]$,
$C[s] \geqslant_B C[t]$ and $C[s] \geqslant_\gwpo C[t]$ by case~\GWPORP{2b(ii)}.
So $C[s] \geqslant_\mgwpo C[t]$.
\end{proof}

We can use the criterion above as follows:
Given a partial status $\pi$ and reduction triples $A$ and $B$
with $({\sqsupseteq_A}, {\sqsupset_A})$ simple w.r.t.\ $\pi$,
we define the partial statuses $\rho$ by
\[
\rho(f) = \{ i \notin \pi(f) \mid \text{the $i$-th argument position of $f$ is invariant w.r.t.\ $A$ and $B$} \}
\]
for each function symbol $f$.
Also, we define $\bar{\rho}(f) = \{ 1, \ldots, n \} \setminus \rho(f)$ for each $n$-ary function symbol $f$.
Then $({\geqslant_\mgwpo}, {>_\gwpo})$ and $({\geqslant_\mgwpo}, {>_\mgwpo})$ induced by $\pi$, $A$ and $B$ are 
$\bar{\rho}$-compatible.

The remaining problem for using \autoref{P:gwpo-invariance} is
how to estimate invariant positions for the underlying reduction triples $A$ and $B$.
For those induced by algebras, this is done by adapting how to do it for reduction pairs,
see for example \cite{SH25}.
For the trivial reduction triple and those induced by precedences,
all argument positions are invariant.
Also, lexicographic combination (\autoref{P:lex-reduction-triple}) preserves invariance:
if an argument position is invariant w.r.t.\ reduction triples $A$ and $B$,
then it is so w.r.t.\ the lexicographic combination of $A$ and $B$.

Lexicographic combination of reduction pairs~\cite{SH25} is another technique to which invariance is relevant;
this can be done by \autoref{P:gwpo-invariance}.
Also, the technique demands normality of all reduction pairs but the last of the combination;
this can be done by using $({\geqslant_\mgwpo}, {>_\mgwpo})$ with \autoref{T:gwpo-monotone-normal}.
Lastly, it demands us to identify \emph{monotone} positions of the strict order $>_\mgwpo$:
Formally, the $i$-th argument position of a function symbol $f$ is monotone w.r.t.\ a relation $\leadsto$
if $f(t_1, \ldots, t_i, \ldots, t_n) \leadsto f(t_1, \ldots, t'_i, \ldots, t_n)$
for all terms $t_1, \ldots, t_n$ and $t'_i$ with $t_i \leadsto t'_i$.
We have a proposition for monotone positions in the style of \autoref{P:gwpo-invariance}.

\begin{prop}
\label{P:gwpo-mono-pos}
Let $>_\mgwpo$ be the MGWPO induced by a partial status $\pi$ and reduction triples 
$({\geqslant_A}, {\sqsupseteq_A}, {\sqsupset_A})$ and $({\geqslant_B}, {\sqsupseteq_B}, {\sqsupset_B})$.
Suppose that $({\sqsupseteq_A}, {\sqsupset_A})$ is simple w.r.t.\ $\pi$.
The $i$-th argument position of a function symbol $f$
is monotone w.r.t.\ $>_\mgwpo$ if $i \in \pi(f)$. \qed
\end{prop}

\subsection{Evaluation with Benchmark Problems}
\label{SS:experiment-reduction-pair}

In order to evaluate GWPOs as reduction pairs, we implemented them
in the prototype termination tool \Mint\footnote{%
The tool and experimental data are available at: \url{https://www.jaist.ac.jp/project/mint/}}
based on the dependency pair framework.
Below we explain the experimental setting, and then discuss the results.

To describe the configuration of reduction pairs in \autoref{tbl:experiment}, we use the notation of \autoref{tbl:notation}.
We additionally put the following assumptions on interpretations.
\begin{itemize}
\item
Linear polynomial interpretations $\AA$ over $\NN$ interpret each function
symbol $f^{(n)}$ as
\[
f_\AA(x_1, \ldots, x_n) = c_0 + c_1 x_1 + \cdots + c_n x_n
\]
where $c_0 \in \NN$ and $c_i \in \{ 0, 1\}$ for all $1 \leqslant i \leqslant n$.
\item
Max/plus interpretations $\AA$ interpret each function symbol $f^{(n)}$ as
\[
f_\AA(x_1, \ldots, x_n) = \max \{ d_0, c_1 x_1 + d_1, \ldots, c_n x_n + d_n \}
\]
where $d_0 \in \NN$, $d_1, \ldots, d_n \in \ZZ$ and $c_1, \ldots, c_n  \in \{ 0, 1 \}$ for all $1 \leqslant i \leqslant n$
\end{itemize}
The restrictions above on coefficients $c_i$ allow us to keep the search for
parameters decidable, as explained in due course.
In contrast, arbitrary interpretations are allowed for $\Finite^\sharp$.
Moreover, by abuse of the notation, $\LinPol$ may denote not only the class of algebras, but
the class of reduction triples $({\geqslant_\AA}, {\geqslant_\AA}, {>_\AA})$
constructed from linear polynomial interpretations $\AA$ (via \autoref{P:redpair-redtriple});
the same goes for $\MaxPlus$.
In this setting, 
$\GWPOtt(\LinPol, \LinPol^\sharp)$ and $\GWPOtt(\LinPol, \MaxPlus^\sharp)$  generalize $\WPOtt(\LinPol)$ as classes of reduction pairs,
and also $\GWPOtt(\MaxPlus, \MaxPlus^\sharp)$ generalizes $\WPOtt(\MaxPlus)$.
These are consequences of \autoref{P:GWPOAA-generalizes-WPOA} and
\autoref{P:GWPOAM-generalizes-WPOA}.

\begin{table}[t]
  \centering
  \caption{Notation for reduction-pair configurations used in the experiment.}
  \label{tbl:notation}
  \begin{tabular}{lll}
    \toprule
    construct & reads as \\
    \midrule
    $\Trivial$ & trivial reduction triple $({\TT \times \TT}, {\TT \times \TT}, {\emptyset})$ \\
    $\Prec$ & precedence as reduction triple (\autoref{P:prec-redtriple}) \\
    $\LinPol$ & linear polynomial interpretation over $\NN$ (as algebra or reduction triple) \\
    $\LinPol^\sharp$ & linear polynomial interpretation over $\NN$ with \autoref{P:algebra-redtriple} \\
    $\MaxPlus$ & max/plus interpretation (as algebra or reduction triple) \\
    $\MaxPlus^\sharp$ & max/plus interpretation with \autoref{P:algebra-redtriple} \\
    $\Finite^\sharp$ & finite algebra over $\{ 0, 1 \}$ with \autoref{P:model-redtriple} \\
    $A \times B$ & lexicographic combination of reduction triples $A$ and $B$ (\autoref{P:lex-reduction-triple}) \\
    $\WPOtt(\AA)$ & WPO induced by algebra $\AA$  (\autoref{T:wpo-reduction-pair}), equals $\GWPOtt(\AA, \Prec)$  \\
    $\GWPOtt(A, B)$ & GWPO induced by reduction triples $A$ and $B$ (\autoref{T:gwpo-reduction-pair}) \\
    $\SPOtt(A)$ & SPO induced by reduction triple $A$ (\autoref{C:mspo-reduction-pair}), equals $\GWPOtt(\Trivial, A)$  \\
    \bottomrule
  \end{tabular}
\end{table}

Our termination tool \Mint is based on the standard techniques in the dependency pair framework,
namely, recursive cycle analysis based on strongly connected components of dependency graphs~\cite{AG00,GAO02,HM05}
and reduction pairs (\autoref{T:rp}).
In addition, we use the usable rule technique (explained in \autoref{SS:rp-refinements})
with tcap~\cite{GTS05FroCoS} and the refinement called usable rule w.r.t.\ argument filter~\cite{GTS+06}.
To handle applicative term rewrite systems more effectively, the tool also uses the uncurrying technique~\cite{HMZ13}.
The search for suitable reduction pairs is performed by encoding the search problem in the language of satisfiability modulo theory (SMT) solvers,
using the standard techniques~\cite{ZHM09,YKS15}.
Our restrictions on coefficients of linear polynomial or max/plus interpretations to $\{ 0, 1 \}$ allow
us to encode it into quantifier-free linear integer arithmetic with if-then-else expressions,
which is known to be decidable.
As discussed at the end of \autoref{SS:rp-refinements},
$\Ce$-compatibility for usable rules is guaranteed for free, except for the case of using finite algebras ($\Finite^\sharp$).
In this case, the reduction pair processor is applied without the usable rule technique.

The problems for the experiment consist of 1542 finite TRSs 
included below the directory \texttt{TRS\_Standard}
of version 11.6 of the Termination Problem Database (TPDB)~\cite{TPDB},
including the secret problems for the Termination Competition 2026.
The termination tool was run on 
a machine equipped with an Intel Core i7-11850H CPU (8 cores, 16 threads, 2.50--4.80 GHz)
and 16 GB RAM, and used Z3 version 4.13.0~\cite{dMB08}
as an SMT solver for constraint solving.
A one-minute timeout was set for each run of a configuration on a termination problem.
So, if a run neither proved (non-)termination nor timed out,
it means that it terminated after exhausting the search space,
or aborted due to some errors such as memory exhaustion.
Some of the problems
contain ill-formed rewrite rules like $x \to \m{f}(y)$.
Our tool reports them as non-terminating TRSs, as other tools do.

Let us explain how to read the result table (\autoref{tbl:experiment}) with an example:
the configuration $\SPOtt(\Finite^\sharp)$
proved termination of 246 TRSs in the problem set,
aborted on 11 problems (actually due to memory exhaustion), and 
reached the time limit for 19 problems. The total runtime on the whole problem set was 60 minutes.

\begin{table}[t]
  \centering
  \caption{Experimental results on 1542 TRSs from TPDB.}
  \label{tbl:experiment}
  \begin{tabular}{lccccr}
    \toprule
    configuration & proofs & timeouts & errors & total (min) & note \\
    \midrule
    $\WPOtt(\LinPol)$ & \emph{578} & 22 & 0 & 50 & equals $\GWPOtt(\LinPol, \Prec)$ \\
    $\WPOtt(\MaxPlus)$ & \emph{616} & 55 & 0 & 101 & equals $\GWPOtt(\MaxPlus, \Prec)$ \\
    $\GWPOtt(\LinPol, \LinPol^\sharp)$ & \emph{607} & 20 & 0 & 44 & generalizes $\WPOtt(\LinPol)$ \\
    $\GWPOtt(\LinPol, \MaxPlus^\sharp)$ & \emph{704} & 45 & 0 & 100 & generalizes $\WPOtt(\LinPol)$ \\
    $\GWPOtt(\MaxPlus, \LinPol^\sharp)$ & \emph{703} & 39 & 0 & 81 &  \\
    $\GWPOtt(\MaxPlus, \MaxPlus^\sharp)$ & \emph{639} & 61 & 0 & 127 & generalizes $\WPOtt(\MaxPlus)$ \\
    $\SPOtt(\LinPol^\sharp \times \LinPol^\sharp)$ & \emph{595} & 22 & 0 & 51 & equals $\GWPOtt(\Trivial, \LinPol^\sharp \times \LinPol^\sharp)$ \\
    $\SPOtt(\LinPol^\sharp \times \MaxPlus^\sharp)$ & \emph{700} & 55 & 0 & 110 & equals $\GWPOtt(\Trivial, \LinPol^\sharp \times \MaxPlus^\sharp)$ \\
    $\SPOtt(\MaxPlus^\sharp \times \LinPol^\sharp)$ & \emph{700} & 60 & 0 & 129 & equals $\GWPOtt(\Trivial, \MaxPlus^\sharp \times \LinPol^\sharp)$ \\
    $\SPOtt(\MaxPlus^\sharp \times \MaxPlus^\sharp)$ & \emph{640} & 79 & 0 & 167 & equals $\GWPOtt(\Trivial, \MaxPlus^\sharp \times \MaxPlus^\sharp)$ \\
    $\SPOtt(\Finite^\sharp)$ & \emph{246} & 19 & 11 & 60 & equals $\GWPOtt(\Trivial, \Finite^\sharp)$ \\
    \midrule
    \AProVE & \emph{1036} & 230 & 0 & 289 & \cite{GAB+17} \\
    \NaTT & \emph{888} & 40 & 1 & 91 & implements WPO~\cite{YKS14} \\
    \bottomrule
  \end{tabular}
\end{table}

Overall, we obtained more termination proofs by switching from the original WPO
to the GWPO\@. Combining different kinds of algebras
is quite effective, which is not possible for the WPO\@.
As in theory, every problem solved by $\WPOtt(\LinPol)$ was also solved by $\GWPOtt(\LinPol, \LinPol^\sharp)$,
and the same holds
for $\WPOtt(\MaxPlus)$ and $\GWPOtt(\MaxPlus, \MaxPlus^\sharp)$.
Due to the time limit, there was just one problem solved by $\WPOtt(\LinPol)$ 
but not by $\GWPOtt(\LinPol, \MaxPlus^\sharp)$.
The configuration $\SPOtt(\Finite^\sharp)$ proved termination of 9 TRSs that could not be proven by the other 10 constructions of reduction pairs.
A representative example is $\{ \m{f}(\m{a}, \m{b}, x) \to \m{f}(x, x, x) \}$ of \autoref{E:Dershowitz} (\texttt{SK90/4.50} in TPDB).

We also discuss the impact of the GWPO on state-of-the-art termination provers.
Due to space limitations, we only consider the top two tools in the \texttt{TRS\_Standard} category of 
the Termination Competition 2026, namely \AProVE~\cite{GAB+17} and \NaTT~\cite{YKS14}.
As shown in \autoref{tbl:experiment},
\AProVE proved termination of 1036 TRSs, and \NaTT proved termination of 888 TRSs, 
in the same experimental setting (with a one-minute time limit).
The configurations of the form $\GWPOtt(*, *)$ or $\SPOtt(*)$
collectively proved termination of 725 distinct TRSs. 
Among them, 
there were 9 TRSs whose termination could not be proved by \NaTT, namely:
\begin{center}
\begin{tabular}{@{}l@{\qquad}l@{}}
\texttt{AProVE\_04/forward\_instantiation}
&
\texttt{Saito\_26/AProVE\_08\_id\_inc\_p}
\\
\texttt{AProVE\_04/forward\_instantiation2}
&
\texttt{Strategy\_removed\_AG01/\#4.25}
\\
\texttt{AProVE\_08/id\_inc}
&
\texttt{Transformed\_CSR\_04/Ex4\_DLMMU04\_iGM}
\\
\texttt{SK90/4.50}
&
\texttt{Various\_04/11}
\\
\texttt{SK90/4.55}
\end{tabular}
\end{center} 
Among them, \AProVE failed to prove termination of
\texttt{Saito\_26/AProVE\_08\_id\_inc\_p}.
The TRS is given below.
\begin{align*}
\m{f}(\m{s}(x)) &\to \m{f}(\idinc(\m{c}(\m{p}(\m{s}(x)), \m{p}(\m{s}(x)))))
&
\idinc(\m{c}(x,y)) &\to \m{c}(\idinc(x),\idinc(y))
&
\m{p}(\m{s}(x))
&\to x
\\
\m{f}(\m{c}(\m{s}(x),y)) &\to \m{g}(\m{c}(\m{p}(\m{s}(x)),y))
&
\idinc(\m{s}(x)) &\to \m{s}(\idinc(\m{p}(\m{s}(x))))
&
\m{p}(\m{0}) &\to \m{0}
\\
\m{g}(\m{c}(\m{s}(x),y)) &\to \m{g}(\m{c}(y,\m{p}(\m{s}(x))))
&
\idinc(\m{0}) &\to \m{0}
\\
\m{g}(\m{c}(x,\m{s}(y))) &\to \m{g}(\m{c}(\m{p}(\m{s}(y)),x))
&
\idinc(\m{0}) &\to \m{s}(\m{0})
\\
\m{g}(\m{c}(x,x)) &\to \m{f}(x)
\end{align*}
This problem was crafted by the first author as a secret problem for the Termination Competition 2026.
There, our tool \Mint  was the only tool that proved the termination.
We refer to \autoref{E:inc} for the intuition behind the termination proof using $\SPOtt(\MaxPlus^\sharp \times \MaxPlus^\sharp)$.

\section{Conclusion}
\label{S:conclusion}

We addressed the issues of the (M)SPO through the lens of the GWPO\@.
In the reduction-order setting, we identified an efficient subclass (\autoref{T:kbo-like-gwpo}),
and also a ground-total subclass (\autoref{T:gwpo-ground-total}) for the GWPO\@.
These results are also applicable to the (M)SPO by appealing to the simulation result (\autoref{T:gwpo-spo} and \autoref{C:mgwpo-mspo}).
In the reduction-pair setting, we devised a reduction-pair version of the GWPO (\autoref{T:gwpo-reduction-pair}),
and thereby the first reduction-pair version of the SPO (\autoref{C:mspo-reduction-pair}).
The experiment with standard benchmarks (\autoref{tbl:experiment})
demonstrated a large improvement over the original WPO.

We conclude with related and future work.
In particular, we briefly discuss some related results presented in the first author's PhD thesis~\cite{Sai26} but not presented here.

Termination modulo associativity and commutativity (AC) is characterized by so-called AC reduction orders
in a similar way to \autoref{P:reduction-order-termination}.
There are already AC versions of KBOs~\cite{Ste90,KV03,YWHM16},
recursive path orders~\cite{Rub02}, multiset path orders~\cite{HM25}, and monotonic semantic path orders~\cite{BR03}.
(In fact, Ochiai, Aoto and Toyama pointed out a flaw in the monotonicity of Borralleras and Rubio's AC-MSPO, and proposed a remedy for it~\cite{OAT05}.)
It is interesting to investigate unified orders for AC rewriting, as has been done for plain rewriting.
As a first step, the first author has investigated an AC version of the GWPO (\cite[Section~3.4]{Sai26})
that corresponds to the E-AC-MSPO of \cite[Section~2.6]{Bor03} (a simpler version of the aforementioned AC-MSPO),
establishing a simulation result.

The WPO has been extended to co-rewrite pairs~\cite{Yam22IJCAR}, which are a generalization of reduction pairs
to prove non-reachability in term rewriting.
Not so surprisingly, the GWPO can be extended to co-rewrite pairs in a similar way;
the details are found in \cite[Chapter~5]{Sai26}.
We note that co-rewrite pairs can be used as discrimination pairs for proving
non-confluence~\cite{Aot13}.

There is an extension of the reduction-pair version of the WPO (found in \cite[Section~4.2]{YKS15})
to properly subsume interpretation-based reduction pairs, i.e., ${\geqslant_\AA} \subseteq {\geqslant_\wpo}$
and ${>_\AA} \subseteq {>_\wpo}$ can be satisfied by choosing a suitable partial status and precedence.
Although this extension does not greatly improve the power of the WPO in practice,
it is natural to investigate a similar extension for the GWPO\@.
In fact, it has turned out that the GWPO with a similar extension is possible,
and can simulate arbitrary non-trivial reduction pairs (\cite[Theorem~4.3.7]{Sai26}).

While the path orders considered here compare arguments only lexicographically,
it is well known that they can be extended to allow choice of
argument comparison method depending on the root function symbols.
In fact, the WPO formalized in Isabelle/HOL~\cite{STY21} extends \autoref{D:wpo-reduction-pair}
so that it allows a choice between multiset and lexicographic extensions
for each symbol, incorporating the well-known extension for the recursive path order.
It is very likely that the same extension is also possible for the GWPO,
because the proofs here seem independent of lexicographic extension.
The general path order (GPO)~\cite{DH95,Ges96} is a unifying framework for such extensions,
parameterizing the way to compare arguments.
The relationship between the GPO and our GWPO is worth investigating.

\section*{Acknowledgment}
We thank Ren\'{e} Thiemann and Akihisa Yamada for
valuable feedback on our work.  We also thank J\"{u}rgen Giesl for
interesting questions at WST 2023 and FroCoS 2023,
which led to \autoref{S:reduction-pair} and \autoref{SS:ground-totality}, respectively.
Special thanks go to Alfons Geser for providing a copy of his technical report~\cite{Ges92}.

\bibliographystyle{alphaurl}
\bibliography{references}

\appendix

\section{Omitted Proofs}
\label{S:omitted}

We show that the termination of \autoref{E:SK} cannot be shown by a WPO as
a reduction order, if the algebra is restricted to a linear polynomial
interpretation.

\begin{prop}
\label{P:S-WPO}
There is no weakly simple linear polynomial interpretation $\AA$ and
precedence such that the induced WPO $>_\wpo$ satisfies 
$(x \cdot z) \cdot (y \cdot z) >_\wpo ((\m{S} \cdot x) \cdot y) \cdot z$.
\end{prop}
\begin{proof}
Let $\AA$ be a weakly simple linear polynomial interpretation of the
form $x \cdot_\AA y = a x + by + c$ and $\m{S}_\AA = s$.  As $\AA$ is
weakly simple, we have $a, b \geqslant 1$.
Suppose
\(
t = (x \cdot z) \cdot (y \cdot z) >_\wpo ((\m{S} \cdot x) \cdot y) \cdot z = u
\).  Consider $t >_\AA u$. It is unfolded to the inequality:
\[
a^2 x + aby + (ab + b^2) z + (a + b + 1) c >
a^2bx + aby + bz + a^3 s + (a^2 + a + 1) c
\]
This leads to $b = 1$ and $(a + b + 1)c > a^3 s + (a^2 + a + 1) c$, yielding a contradiction.
So case~\WPO{1} is impossible.
Moreover, cases~\WPO{2a} and \WPO{2b(ii)} are impossible because of the
variables, and case~\WPO{2b(i)} is also impossible because $t$ and $u$ have
the same root symbols.
\end{proof}

Next, we prove \autoref{T:spo}.
Let $({\sqsupseteq}, {\sqsupset})$ be a stable order pair on $\TT$ and
$>_\spo$ the induced SPO.
The transitivity, irreflexivity, and stability of $>_\spo$ are straightforward to prove.

\begin{lem}
\label{L:spo}
The SPO $>_\spo$ is a stable strict order with the subterm property. \qed
\end{lem}

It remains to prove well-foundedness of $>_\spo$
under the assumption that $\sqsupset$ is well-founded and $\FF$ is bounded.
This follows from 
\autoref{C:mspo-reduction-pair} and \autoref{R:ro-rp}:
$>_\spo$ is well-founded
because it is contained in the well-founded order of the 
reduction-pair version of the SPO.

We present a proof of \autoref{T:gwpo-reduction-pair}
following the proofs for the WPO~\cite{YKS15,Yam22IJCAR}.
Except for well-foundedness, we only indicate the proof strategies by induction
because the rest is routine;
see \cite[Chapter~4]{Sai26} for the omitted details.
Let $({\geqslant_\gwpo}, {>_\gwpo})$ be the GWPO induced by a partial status $\pi$
and pairs $({\sqsupseteq_A}, {\sqsupset_A})$ and $({\sqsupseteq_B}, {\sqsupset_B})$
of relations on terms.

\begin{lem}
\label{L:gwpo-normal}
The inclusion ${>_\gwpo} \subseteq {\geqslant_\gwpo}$ holds.
\end{lem}
\begin{proof}
Easy induction on $|s| + |t|$ shows that $s >_\gwpo t$ implies $s \geqslant_\gwpo t$.
\end{proof}

\begin{lem}
\label{L:gwpo-compatibility}
If $({\sqsupseteq_A}, {\sqsupset_A})$ and $({\sqsupseteq_B}, {\sqsupset_B})$ are order pairs,
then $\geqslant_\gwpo$ and $>_\gwpo$ are transitive, and $({\geqslant_\gwpo}, {>_\gwpo})$ is compatible.
\end{lem}
\begin{proof}
As in \cite[Lemma~7]{YKS15},
the claim is proven by showing the next four properties for all terms $s, t$ and $u$:
if $s \geqslant_\gwpo t$ and $t \geqslant_\gwpo u$ then $s \geqslant_\gwpo u$;
if $s \geqslant_\gwpo t$ and $t >_\gwpo u$ then $s >_\gwpo u$;
if $s >_\gwpo t$ and $t \geqslant_\gwpo u$ then $s >_\gwpo u$; and
if $s >_\gwpo t$ and $t >_\gwpo u$ then $s >_\gwpo u$.
This is done by easy simultaneous induction on $|s| + |t| + |u|$.
\end{proof}

\begin{lem}
\label{L:gwpo-refl-pi-subterm}
If $\sqsupseteq_A$ and $\sqsupseteq_B$ are reflexive and $A$ is simple with respect to $\pi$,
then
\begin{itemize}
\item $t >_\gwpo t_i$ whenever $i \in \pi(f)$ and $t = f(\seq[n]{t})$, and
\item $t \geqslant_\gwpo t$ for all $t$ (i.e., $\geqslant_\gwpo$ is reflexive).
\end{itemize}
\end{lem}
\begin{proof}
As in \cite[Lemma~1]{Yam22IJCAR},
the claims are shown by simultaneous induction on $|t|$.
\end{proof}

Next we show stability via \autoref{L:gwpo-refl-pi-subterm}.

\begin{lem}
\label{L:gwpo-stable}
If $({\sqsupseteq_A}, {\sqsupset_A})$ and $({\sqsupseteq_B}, {\sqsupset_B})$ are stable,
$\sqsupseteq_A$ and $\sqsupseteq_B$ are reflexive, and
$A$ is simple with respect to $\pi$,
then $({\geqslant_\gwpo}, {>_\gwpo})$ is stable. 
\end{lem}
\begin{proof}
It suffices to prove the following properties: for all terms $s$ and $t$ and substitutions $\sigma$,
if $s \geqslant_\gwpo t$ then $s\sigma \geqslant_\gwpo t\sigma$, and
if $s >_\gwpo t$ then $s\sigma >_\gwpo t\sigma$.
Following \cite[Lemma~2]{Yam22IJCAR},
we can prove these claims simultaneously by induction on $|s| + |t|$.
\end{proof}

For our well-foundedness proof,
we recall the notion of \emph{lexicographic product}.
In contrast to lexicographic combination, order pairs on possibly different sets are combined 
into a strict order on the set of tuples.

\begin{defi}
\label{D:lex-prod}
Let $({\geqslant_1},{>_1}), \ldots ,({\geqslant_n},{>_n})$ be order pairs
on sets $A_1, \ldots, A_n$, respectively.
The \emph{lexicographic product}
\(
({\geqslant_1},{>_1}) \lexprod \cdots \lexprod ({\geqslant_n},{>_n})
\)
is the strict order $>$ defined on $A_1 \times \cdots \times A_n$ as
follows: $(\seq{a}) > (\seq{b})$ if there exists an index 
$k \in \{ 1, \ldots, n \}$ such that $a_k >_k b_k$ and
$a_j \geqslant_j b_j$ for all $1 \leqslant j < k$.
\end{defi}

Like lexicographic combination, the product operation preserves well-foundedness.

\begin{prop}
\label{P:lex-prod}
The lexicographic product of order pairs  $({\geqslant_1},{>_1}), \ldots ,({\geqslant_n},{>_n})$
is well-founded if every $>_i$ is well-founded.
\end{prop}

For showing that $>_\gwpo$ is well-founded, we adopt Buchholz's method
\cite{Buc95}.  One can find a similar proof in \cite[Lemma 8]{YKS15}.  We
write $\SN(>_\gwpo)$ for the set of all terms $t$ such that there is no
infinite descending sequence
$t >_\gwpo t_1 >_\gwpo t_2 >_\gwpo \cdots$ starting from $t$.\footnote{SN
stands for strong normalization, which is another name for termination.}
We have the following basic properties:
\begin{itemize}
    \item The restriction of $>_\gwpo$ to $\SN(>_\gwpo)$ is a well-founded order on $\SN(>_\gwpo)$.
    \item $t \in \SN(>_\gwpo)$ if $u \in \SN(>_\gwpo)$ for all terms $u$ with $t >_\gwpo u$.
\end{itemize}
To exploit these properties,
we need the assumptions to use \autoref{L:gwpo-compatibility}.

\begin{lem}
\label{L:buchholz-gwpo}
Suppose that
$({\sqsupseteq_A}, {\sqsupset_A})$ and $({\sqsupseteq_B}, {\sqsupset_B})$ are well-founded stable order pairs, and
$A$ is simple with respect to $\pi$.
Then $t \in \SN(>_\gwpo)$ whenever $s = f(\seq[m]{s}) >_\gwpo t$
and $s_i \in \SN(>_\gwpo)$ for all $i \in \pi(f)$.
\end{lem}
\begin{proof}
We perform induction on the quadruple $(s, s, \pi(f)(\seq[m]{s}), t)$
with respect to $\gg$, where $\gg$ is the lexicographic product
\[
({\sqsupseteq_A}, {\sqsupset_A}) \lexprod
({\sqsupseteq_B}, {\sqsupset_B}) \lexprod
({\geqslant_\gwpo^\lex}, {>_\gwpo^\lex}) \lexprod
({\superterm}, {\prsuperterm})
\]
defined on $\TT \times \TT \times \SN(>_\gwpo)^{\leqslant k} \times \TT$ for an upper bound $k$ of arity.
It suffices to prove that $u \in \SN(>_\gwpo)$ whenever $t >_\gwpo u$.
If $t$ is a variable, then by definition $t \sqsupset_A u$ must hold.
By applying $\sigma = \{ t \mapsto u \}$ to $t \sqsupset_A u$ successively, we obtain
$u \sqsupset_A u\sigma \sqsupset_A (u\sigma)\sigma \sqsupset_A \cdots$, which contradicts the well-foundedness of $\sqsupset_A$.
So, letting $t = g(\seq[n]{t})$, we analyze the derivation of $s >_\gwpo t$.
\begin{enumerate}[1.]
\item Suppose $s \sqsupset_A t$.
By the simplicity assumption, we have $s \sqsupset_A t \sqsupseteq_A t_j$ or $s \sqsupset_A t \sqsupset_A t_j$,
and therefore  for all $j \in \pi(g)$ we have $s >_\gwpo t_j$ by case~\GWPORP{1}.
So, we apply the induction hypothesis with $(s, s, \pi(f)(\seq[m]{s}), \underline{t}) \gg (s, s, \pi(f)(\seq[m]{s}), \underline{t_j})$
to obtain $t_j \in \SN(>_\gwpo)$ for all $j \in \pi(g)$.
Here, for readability, the decreasing components are indicated by underlining.
Then, we apply the induction hypothesis with $(\underline{s}, s, \pi(f)(\seq[m]{s}), t) \gg (\underline{t}, t, \pi(g)(\seq[n]{t}), u)$
to conclude $u \in \SN(>_\gwpo)$.
\item Suppose $s \sqsupseteq_A t$.
\begin{enumerate}[a.]
    \item Assume  $s = f(\seq[m]{s})$ and $s_i \geqslant_\gwpo t$ for some $i \in \pi(f)$.
    Then we have $s_i >_\gwpo u$ by \autoref{L:gwpo-compatibility}, so $s_i \in \SN(>_\gwpo)$ yields $u \in \SN(>_\gwpo)$.
    \item Assume $s >_\gwpo t_j$ for all $j \in \pi(g)$.
    We apply the induction hypothesis with 
    \[
    (s, s, \pi(f)(\seq[m]{s}), \underline{t}) \gg (s, s, \pi(f)(\seq[m]{s}), \underline{t_j})
    \]
    to obtain $t_j \in {\SN(>_\gwpo)}$ for all $j \in \pi(g)$.
    Now, we apply the induction hypothesis with
    \[
        (s, s, \pi(f)(\seq[m]{s}), t) \gg (t, t, \pi(g)(\seq[n]{t}), u)
    \]
    to obtain $u \in \SN(>_\gwpo)$, where the inequality of $\gg$ is justified as follows: 
    if case~\GWPORP{2b(i)} holds, 
    the second component decreases;
    if case~\GWPORP{2b(ii)} holds, 
    the third decreases. \qedhere
\end{enumerate}
\end{enumerate}
\end{proof}

\begin{lem}
\label{L:gwpo-wf}
Suppose that
$({\sqsupseteq_A}, {\sqsupset_A})$ and $({\sqsupseteq_B}, {\sqsupset_B})$ are well-founded stable order pairs, and
$A$ is simple with respect to $\pi$.
Then $>_\gwpo$ is well-founded.
\end{lem}
\begin{proof}
We show that $t \in \SN(>_\gwpo)$ for all terms $t$ by structural induction on $t$.
It suffices to show that $u \in \SN(>_\gwpo)$ whenever $t >_\gwpo u$.
By the same argument used in the proof of \autoref{L:buchholz-gwpo},
we can assume that $t$ is not a variable, so let $t = f(\seq[n]{t})$.
By the induction hypothesis $t_i \in \SN(>_\gwpo)$ for all $i$,
so \autoref{L:buchholz-gwpo} applies.
\end{proof}

To conclude \autoref{T:gwpo-reduction-pair}, it remains to prove monotonicity.
Let $({\geqslant_\mgwpo}, {>_\mgwpo})$ be the monotone generalized weighted path order
induced by a partial status $\pi$ and triples $({\geqslant_A}, {\sqsupseteq_A}, {\sqsupset_A})$ and $({\geqslant_B}, {\sqsupseteq_B}, {\sqsupset_B})$
of relations on terms.

\begin{lem}
\label{L:gwpo-monotone}
Suppose that 
$({\geqslant_A}, {\sqsupseteq_A}, {\sqsupset_A})$ and $({\geqslant_B}, {\sqsupseteq_B}, {\sqsupset_B})$ are reduction triples, and
$({\sqsupseteq_A}, {\sqsupset_A})$ is simple with respect to $\pi$.
Then $\geqslant_\mgwpo$ is monotone.
If in addition $\pi$ is total, then $>_\mgwpo$ is also monotone.
\end{lem}
\begin{proof}
The claim can be shown using \autoref{L:gwpo-refl-pi-subterm}, following the line of \autoref{T:mspo}.
\end{proof}

In \autoref{tbl:experiment} it is claimed that $\GWPOtt(\LinPol, \LinPol^\sharp)$ generalizes $\WPOtt(\LinPol)$
and $\GWPOtt(\MaxPlus, \MaxPlus^\sharp)$ generalizes $\WPOtt(\MaxPlus)$, as classes of reduction pairs.
To show this, let $\FF$ be the finite signature, and $({\geqslant_\wpo}, {>_\wpo})$ the WPO induced by an $\FF$-algebra $\AA$ on $\NN$ (ordered as usual),
a precedence $\succeq$ and a partial status $\pi$.
The signature $\FF$ can be partitioned into equivalence classes
$\seq{F}$ induced by $\succeq$.  We can assume that they are topologically
sorted. So we have $i > j$ whenever $f \in F_i$, $g \in F_j$, and $f \succ g$.
Now we extend $\AA$ to the $(\FF \cup \FF^\sharp)$-algebra $\AA^\sharp$
with $f_\AA^\sharp(\seq[n]{x}) = k$ for each $f^{(n)} \in F_k$.
Then we have:

\begin{prop}
\label{P:GWPOAA-generalizes-WPOA}
The GWPO $({\geqslant_\mgwpo}, {>_\gwpo})$ induced by the partial status $\pi$
and the reduction triples $({\geqslant_\AA}, {\geqslant_\AA}, {>_\AA})$ and $({\geqslant_\AA}, {\geqslant_\AA^\sharp}, {>_\AA^\sharp})$
extends  $({\geqslant_\wpo}, {>_\wpo})$.
\end{prop}
\begin{proof}
By construction, $({\geqslant_\gwpo}, {>_\gwpo})$ extends $({\geqslant_\wpo}, {>_\wpo})$.
Moreover, since $\AA$ is normal, 
${\geqslant_\AA}$ extends ${\geqslant_\wpo}$.
The claim immediately follows from these facts.
\end{proof}

By slightly changing the argument above, we can show that  
$\GWPOtt(\LinPol, \MaxPlus^\sharp)$ generalizes $\WPOtt(\LinPol)$:
After taking the partition $\seq{F}$, we define a max/plus $(\FF \cup \FF^\sharp)$-algebra $\BB$ by
\begin{align*}
f_\BB(\seq[n]{x}) &= \max \{ 0, x_1, \ldots, x_n \}
&
f^\sharp_\BB(\seq[n]{x}) &= k
\end{align*}
for each $f^{(n)} \in F_k$. Then we have:
\begin{prop}
\label{P:GWPOAM-generalizes-WPOA}
Let $({\geqslant_\mgwpo}, {>_\gwpo})$ be the GWPO induced by the partial status $\pi$
and the reduction triples $({\geqslant_\AA}, {\geqslant_\AA}, {>_\AA})$ and $({\geqslant_\BB}, {\geqslant_\BB^\sharp}, {>_\BB^\sharp})$,
and let $s$ and $t$ be terms with $\Var(s) \supseteq \Var(t)$.
Then $s \geqslant_\wpo t$ implies $s \geqslant_\mgwpo t$,
and also $s >_\wpo t$ implies $s >_\gwpo t$.
\end{prop}
\begin{proof}
By construction,
$({\geqslant_\gwpo}, {>_\gwpo})$ extends $({\geqslant_\wpo}, {>_\wpo})$,
and ${\geqslant_\AA}$ extends ${\geqslant_\wpo}$,
as in the proof of \autoref{P:GWPOAA-generalizes-WPOA}.
Moreover, $s \geqslant_\BB t$ 
follows from $\Var(s) \supseteq \Var(t)$.
So the claim holds.
\end{proof}

Since the dependency pair framework assumes TRSs with the variable conditions,
the requirement $\Var(s) \supseteq \Var(t)$ in the statement is satisfied for free.

\end{document}